\documentclass[12pt]{article}

\usepackage[T1]{fontenc}
\usepackage[utf8]{inputenc}
\DeclareUnicodeCharacter{0300}{\`}
\DeclareUnicodeCharacter{0301}{\'}
\DeclareUnicodeCharacter{0302}{\^}
\DeclareUnicodeCharacter{0303}{\~}
\DeclareUnicodeCharacter{0308}{\"}
\DeclareUnicodeCharacter{0327}{\c}
\usepackage[a4paper,margin=1.02in]{geometry}
\usepackage{microtype}
\usepackage{setspace}
\usepackage{amsmath,amssymb,amsthm,mathtools}
\usepackage{bm}
\usepackage{booktabs}
\usepackage{threeparttable}
\usepackage{tabularx}
\usepackage{adjustbox}
\usepackage{enumitem}
\usepackage{graphicx}
\usepackage{caption}
\usepackage{float}
\usepackage{placeins}
\usepackage{hyperref}
\usepackage[nameinlink,noabbrev]{cleveref}
\usepackage[backend=biber,style=authoryear,natbib=true,maxcitenames=2,maxbibnames=99]{biblatex}
\graphicspath{{./}{figures/}}
\hypersetup{
  hidelinks,
  pdfauthor={Miquel Noguer i Alonso and Ali Al-Fallouji},
  pdftitle={Tail Risk Management with Puts and Trend Following: A CVaR Framework for Crashes and Drawdowns}
}
\newtheorem{theorem}{Theorem}[section]
\newtheorem{proposition}[theorem]{Proposition}
\newtheorem{lemma}[theorem]{Lemma}
\newtheorem{corollary}[theorem]{Corollary}
\theoremstyle{definition}

\newtheorem{assumption}[theorem]{Assumption}
\theoremstyle{remark}
\newtheorem{remark}[theorem]{Remark}

\crefname{proposition}{Proposition}{Propositions}
\Crefname{proposition}{Proposition}{Propositions}
\crefname{lemma}{Lemma}{Lemmas}
\Crefname{lemma}{Lemma}{Lemmas}
\crefname{corollary}{Corollary}{Corollaries}
\Crefname{corollary}{Corollary}{Corollaries}
\crefname{assumption}{Assumption}{Assumptions}
\Crefname{assumption}{Assumption}{Assumptions}
\crefname{definition}{Definition}{Definitions}
\Crefname{definition}{Definition}{Definitions}
\crefname{remark}{Remark}{Remarks}
\Crefname{remark}{Remark}{Remarks}

\newcommand{\E}{\mathbb{E}}
\newcommand{\Prob}{\mathbb{P}}
\newcommand{\R}{\mathbb{R}}
\newcommand{\Q}{\mathbb{Q}}
\newcommand{\F}{\mathcal{F}}
\newcommand{\CVaR}{\mathrm{CVaR}}
\newcommand{\VaR}{\mathrm{VaR}}

\title{\vspace{-3.0em}Tail Risk Management with Puts and Trend Following:\\
A CVaR Framework for Crashes and Drawdowns}

\author{%
\begin{tabular}{c@{\hspace{4em}}c}
Miquel Noguer i Alonso & Ali Al-Fallouji\\[-0.05em]
{\small Artificial Intelligence Finance Institute (AIFI)} & {\small Mirabaud Group}
\end{tabular}}

\date{July 1, 2026}

\begin{document}
\maketitle
\vspace{-1.6em}

\begin{abstract}
Tail-risk management is not only an instrument-selection problem. It is an
allocation problem across loss mechanisms: abrupt crash states, volatility
repricing, and persistent drawdowns require different forms of protection. This
paper develops a continuous-time CVaR framework that places two common
protection sleeves---long out-of-the-money put options and systematic
trend-following overlays---inside one coherent tail-risk mandate. The option
sleeve is modeled as a marked-to-market traded asset, so premium drag, diffusion
exposure, and jump repricing enter through its physical return process rather
than through inconsistent terminal-payoff accounting. The resulting Markov state
contains wealth, spot, stochastic variance, and an exponentially weighted
log-return signal, and we derive the associated Hamilton--Jacobi--Bellman
equation in viscosity form. The main analytical separation is temporal: convex
insurance reprices immediately on jump impact, whereas trend following is late
on the first shock because its signal must cross zero, but becomes increasingly
defensive during persistent drawdowns without requiring fresh option premium. We
then give sufficient and local conditions for an interior hybrid allocation,
derive a CVaR policy-gradient identity, and introduce a four-axis diagnostic
layer separating conditional convexity, tail-event reliability, non-stress carry,
and drawdown persistence. Stylized Monte Carlo experiments illustrate the
mechanism: fixed equal-weight hybrids and grid-optimized hybrids reduce terminal
CVaR relative to either pure sleeve in the reported regimes, while the exact
weight location remains calibration-dependent. The contribution is a transparent
risk-management framework for deciding how much convex crash protection and how
much signal-driven drawdown protection a mandate should hold.
\end{abstract}

\vspace{0.4em}

\newpage

\section{Introduction}

Institutions that care about large drawdowns do not face a generic portfolio
problem. They face a state-contingent downside problem: losses arrive through a
mixture of abrupt jumps, volatility spikes, and persistent directional selloffs.
That is exactly where classical mean-variance analysis is weakest. It treats
upside and downside dispersion symmetrically and gives no special role to the
shape of the left tail \citep{Markowitz1952}.

We therefore use the phrase \emph{tail-risk management} deliberately. A hedge is
an instrument-level object; a tail-risk mandate is a portfolio-level design
problem. The mandate must decide how much convexity, reliability, carry drag,
and persistence it is willing to hold in order to reduce expected loss in the
left tail. Put options and trend following are two sleeves inside that broader
allocation problem, not mutually exclusive answers to the same question.

In practice, two very different hedging technologies dominate that left-tail
problem. The first is explicit convex insurance through long out-of-the-money
index puts. Their appeal is obvious: if the market gaps down, the option
reprices immediately and contractually. The second is dynamic trend-following,
implemented through signal-driven futures or overlay trades that progressively
cut or reverse exposure as price action deteriorates. Trend is slower, but it
can be positively carried over long samples and performs especially well when
losses unfold over time rather than on impact \citep{Hurst2017,Ilmanen2020}.

Those facts suggest a simple economic principle: the best hedge depends on the
transmission mechanism of the drawdown. Puts are strongest against sudden crash
states, but systematic option rolling often carries a persistent premium drag
because index implied variance tends to exceed subsequently realized variance
\citep{Broadie2009,Israelov2017}. Trend has the opposite profile. It is not
contractual and can be caught long during violent reversals, yet it often works
well in prolonged selloffs because it converts persistent negative returns into
progressively more defensive positioning \citep{Moskowitz2012,Hurst2017}.

The systematic hedging literature also suggests that hedge quality is not a
one-dimensional object. It is often decomposed into conditional convexity or
asymmetric co-movement with the market in bad states, reliability or the
probability of producing a positive payoff in tail events, carry outside stress,
and persistence across drawdown horizons. Downside-risk models emphasize that
market exposures can change precisely when the price of risk is high
\citep{Lettau2014}; asset-pricing theory puts the cost of carry at the center of
expected-return trade-offs \citep{Cochrane2005}; and the managed-futures and
time-series-momentum evidence highlights the importance of drawdown duration
\citep{Moskowitz2012,Hurst2013,Hurst2017}. A scalar CVaR criterion aggregates
these mechanisms. That aggregation is useful for optimization, but it can hide
which economic channel is actually providing protection.

What is still missing is a single framework that puts both channels inside the
same tail-risk objective with internally consistent accounting while also making
these hedge-quality dimensions explicit. This paper provides that framework. We
model the risky asset with Heston-type stochastic variance and negative jumps,
let the investor choose both directional exposure and convex-overlay exposure,
and evaluate the resulting strategy through terminal-loss CVaR using the
Rockafellar--Uryasev representation
\citep{Rockafellar2000,Acerbi2002,Artzner1999}. The resulting control problem
produces five concrete improvements over existing heuristic discussions.

\paragraph{Contributions.}
The paper makes five contributions.
\begin{enumerate}[leftmargin=1.6em]
    \item It formulates convex insurance and dynamic trend inside one
    continuous-time CVaR-minimization problem with a fully Markov state.
    Crucially, the option overlay is handled as a traded mark-to-market asset,
    which resolves the premium-accounting inconsistency that often appears in
    terminal-payoff comparisons.
    \item It makes explicit a four-axis hedge-quality diagnostic---conditional
    convexity, tail-event reliability, non-stress carry, and drawdown
    persistence---and shows how these dimensions can enter the problem either as
    constraints or as penalty terms around the CVaR objective.
    \item It derives the jump-diffusion HJB equation in viscosity form for the
    augmented state consisting of wealth, spot, variance, and the trend signal.
    \item It gives analytical comparative statics that separate the two hedge
    channels: puts hedge jump states on impact, whereas trend hedges drawdowns
    that last long enough for the signal to cross through zero and turn the
    portfolio defensive.
    \item It derives a CVaR policy-gradient identity and connects it to a fast
    simulation pipeline used in our numerical experiments, which delivers
    stylized but informative regime evidence.
\end{enumerate}

\paragraph{Scope.}
This is a theory-first paper with stylized Monte Carlo evidence. It is not a
full historical backtest, and it does not claim a complete option-surface
calibration. In particular, the experiments use a fast Black--Scholes proxy for
option valuation inside the simulation section. That choice is appropriate for
illustrating the economics of the hybrid hedge, but it should not be confused
with a production empirical implementation.

The empirical claims are therefore stated in the language of \emph{mechanism
validation}: the simulations check whether the model reproduces the expected
ranking of protection channels across crash, prolonged-bear, and volatility
states. They do not establish universal dominance of the hybrid over every
possible option program or trend specification. The reported hybrid optima are
best read as calibrated examples of the allocation logic developed below.

The four-axis diagnostic layer should be read in the same spirit. It is not yet
a full historical taxonomy of all tradable hedges. It is a transparent way to
label the mechanisms that the CVaR objective otherwise compresses into a single
number, and it provides a bridge from the two-channel theoretical model to a
larger structuring universe.

The analysis draws on continuous-time portfolio choice
\citep{Merton1969,Merton1971}, stochastic control for jump diffusions
\citep{Oksendal2005,FlemingSoner2006,Pham2009}, the economics of option-based
crash protection \citep{Bhansali2014,Broadie2009,Israelov2017}, and the
empirical literature on trend-following and time-series momentum
\citep{Moskowitz2012,Hurst2013,Hurst2017,Ilmanen2020}. The rest of the paper is organized
as follows. \Cref{sec:setup} introduces the market environment and the two hedge
channels. \Cref{sec:cvar} states the CVaR objective. \Cref{sec:control} gives
the dynamic control problem and its HJB characterization. \Cref{sec:comparative}
contains the key comparative statics and a reduced-form interior-hybrid result.
\Cref{sec:gradient} derives the policy-gradient identity. \Cref{sec:mc}
reports stylized Monte Carlo evidence from our numerical experiments.
\Cref{sec:discussion} and \Cref{sec:conclusion} close.

\section{Economic setup}
\label{sec:setup}

\subsection{Asset dynamics}

Let $(\Omega,\F,\{\F_t\}_{t\in[0,T]},\Prob)$ be a filtered probability space
satisfying the usual conditions. The risky asset price $S_t$ follows a
stochastic-volatility jump diffusion,
\begin{equation}
\frac{dS_t}{S_{t^-}} = \mu\,dt + \sqrt{v_t}\,dW_t^S + (e^{\zeta}-1)\,dN_t,
\label{eq:S}
\end{equation}
where $N_t$ is a Poisson process with intensity $\lambda>0$ and i.i.d.
log-jump sizes $\zeta$ with distribution $F_\zeta$ supported on a bounded
interval
\begin{equation}
\label{eq:zeta_bounded}
\zeta \in [\underline\zeta,\overline\zeta]\subset\R,
\qquad \underline\zeta<0<\overline\zeta.
\end{equation}
Boundedness of the jump support is used below to guarantee strict wealth
positivity under two-sided directional exposure. In practice $F_\zeta$ can be
taken as a truncated Gaussian $\mathcal{N}(\mu_\zeta,\sigma_\zeta^2)$
restricted to $[\underline\zeta,\overline\zeta]$; the truncation bounds are
non-binding in the stylized calibration of \Cref{tab:params} because
$\overline\zeta$ is set at several standard deviations above $\mu_\zeta$.
The variance process is of Heston type \citep{Heston1993}:
\begin{equation}
 dv_t = \kappa(\theta-v_t)dt + \xi\sqrt{v_t}\,dZ_t,
 \qquad dW_t^S\,dZ_t = \rho\,dt,
 \label{eq:v}
\end{equation}
with $\kappa,\theta,\xi>0$ and $\rho\in(-1,0)$. The negative correlation
captures the leverage effect. Throughout we assume the Feller condition
$2\kappa\theta>\xi^2$ holds, which keeps $v_t$ strictly positive almost surely
and simplifies the boundary analysis of the integro-PDE below.

For signal construction it is convenient to work with log-prices. Writing
$y_t := \log S_t$, It\^o's formula for jump diffusions yields
\begin{equation}
 dy_t = \Bigl(\mu - \tfrac12 v_t\Bigr)dt + \sqrt{v_t}\,dW_t^S + \zeta\,dN_t.
 \label{eq:logS}
\end{equation}

\begin{assumption}[Admissible controls]
\label{ass:admissible}
The investor chooses a progressively measurable pair
$c_t=(\pi_t,q_t)$ taking values in the compact set
$\mathcal{K}=[\underline\pi,\overline\pi]\times[0,\overline q]$, where
$\pi_t$ is the directional exposure to the risky asset and $q_t$ is the return
exposure to the convex overlay. Directional exposure is two-sided:
$\underline\pi<0\le \overline\pi$. To guarantee strict wealth positivity
across both directions of jump, the directional bounds and the jump support
in \eqref{eq:zeta_bounded} are required to satisfy the two-sided leverage
conditions
\begin{equation}
\label{eq:leverage_cond}
\overline\pi+\overline q < 1,
\qquad
1+\underline\pi\,(e^{\overline\zeta}-1)-\overline q>0,
\end{equation}
together with the overlay bound $R_P\ge -1$ used in Lemma~\ref{lem:positivity}
below. The first condition controls the worst-case \emph{downside} jump for a
long position (where $e^{\underline\zeta}-1\ge -1$ and $R_P\ge -1$); the
second controls the worst-case \emph{upside} jump for a short position.
Controls satisfy $\E[\int_0^T (\pi_t^2+q_t^2)dt]<\infty$.
\end{assumption}

A risk-free asset is available, but we work in excess-return units and set the
short rate to zero throughout.

\subsection{Convex insurance as a traded overlay}
\label{sec:overlay}

We need to distinguish two economic objects that are often conflated in the
structured-hedging literature but which have different mathematical
structure: a \emph{fixed-maturity} traded option, and a \emph{rolled}
constant-maturity trading strategy. The continuous-time theory below works
with the former; the Monte Carlo experiments in \Cref{sec:mc} implement the
latter. The two are related through a simple roll-date correction that we
make explicit at the end of this subsection.

\paragraph{The fixed-maturity overlay.}
Let $T_{\mathrm{opt}}\in(0,T]$ be a fixed option maturity and let
$P(t,s,v)$ denote the marked-to-market price, at time $t<T_{\mathrm{opt}}$
and state $(S_t,v_t)=(s,v)$, of a European option with payoff
$\Phi(S_{T_{\mathrm{opt}}})$ at maturity. In the put case of interest,
$\Phi(S)=(K-S)^+$ with $K\in(0,\infty)$. The overlay is a genuine traded
asset between inception and maturity, and the investor can hold fractional
positions of notional return exposure $q_t$ to its physical return process.

\begin{assumption}[Overlay regularity]
\label{ass:overlay}
There exists a measure $\Q$, equivalent to $\Prob$ on $\mathcal{F}_T$,
under which $P(t,S_t,v_t)$ is the discounted martingale price of the option
payoff $\Phi(S_{T_{\mathrm{opt}}})$. The map
$P:[0,T_{\mathrm{opt}}]\times(0,\infty)^2\to(0,\infty)$ is once continuously
differentiable in time and twice continuously differentiable in $(s,v)$
on $[0,T_{\mathrm{opt}})\times(0,\infty)^2$, with locally bounded derivatives,
and $P(t,s,v)\ge (K-s)^+\ge 0$. The overlay's physical return process is
obtained by applying It\^o's formula to $P(t,S_t,v_t)$ under $\Prob$.
\end{assumption}

Define the option elasticity and vega loading by
\begin{equation}
\delta_P(t,s,v):=\frac{sP_s(t,s,v)}{P(t,s,v)},
\qquad
\nu_P(t,s,v):=\frac{P_v(t,s,v)}{P(t,s,v)}.
\label{eq:elasticity}
\end{equation}
Under Assumption~\ref{ass:overlay}, the overlay return can be written as
\begin{equation}
\frac{dP_t}{P_{t^-}}
= \mu_P(t,S_t,v_t)dt
+ \delta_P(t,S_t,v_t)\sqrt{v_t}\,dW_t^S
+ \nu_P(t,S_t,v_t)\xi\sqrt{v_t}\,dZ_t
+ R_P(t,S_{t^-},v_t;\zeta)\,dN_t,
\label{eq:put_return}
\end{equation}
where
\begin{equation}
R_P(t,s,v;\zeta):=\frac{P(t,se^{\zeta},v)-P(t,s,v)}{P(t,s,v)}.
\label{eq:RPdef}
\end{equation}
The physical drift collects the remaining terms from It\^o's formula:
\begin{equation}
\mu_P(t,s,v) := \frac{1}{P}\Bigl(P_t + \mu s P_s + \kappa(\theta-v)P_v
+ \tfrac12 v s^2 P_{ss} + \rho\xi v s P_{sv} + \tfrac12 \xi^2 v P_{vv}\Bigr).
\label{eq:muP}
\end{equation}

This convention writes the jump contribution in raw $dN_t$ form. Thus
$\mu_P$ is the continuous drift of the overlay between jump arrivals. The
unconditional instantaneous physical drift of $P$ is
$\mu_P+\lambda\E_\zeta[R_P(t,s,v;\zeta)]$. Equivalently, one could rewrite
\eqref{eq:put_return} with the compensated process
$d\widetilde N_t=dN_t-\lambda dt$ and absorb the compensator into the drift.
We keep the raw-jump convention because it makes crash-state repricing explicit
in the wealth equation and in the HJB jump operator.

\begin{proposition}[Marked-to-market option decomposition]
\label{prop:option_decomp}
Under Assumption~\ref{ass:overlay}, the dynamics in \eqref{eq:put_return} hold. The jump
return is the immediate percentage repricing of the overlay when the underlying
moves from $s$ to $se^{\zeta}$.
\end{proposition}

\begin{proof}
Apply It\^o's formula for jump diffusions to the semimartingale
$P(t,S_t,v_t)$. The continuous part produces the drift and the two diffusive
loadings. The jump part is exactly the mark-to-market repricing
$P(t,S_{t^-}e^{\zeta},v_t)-P(t,S_{t^-},v_t)$ divided by the pre-jump price.
\end{proof}

\begin{remark}[Where carry enters]
\label{rem:carry}
The option premium is not ignored anywhere in this formulation. It is encoded in
$P_t$ and therefore in the overlay's physical drift $\mu_P$. Historical evidence
suggests that this drift is often negative for systematically rolled index puts
because of the variance risk premium \citep{Broadie2009,Israelov2017}. The
continuous-time theory below needs only the physical return process in
\eqref{eq:put_return}; it does not require the premium to be modeled separately
at terminal time.
\end{remark}

\paragraph{The rolled-strategy implementation.}
In the Monte Carlo experiments of \Cref{sec:mc} the overlay is implemented
as a self-financing trading rule rather than a single fixed-maturity
contract. At each roll date $\tau_k = k\Delta$, $k=0,1,\ldots$, the
investor simultaneously (i) liquidates the previous contract at its
prevailing market price $P^{(k-1)}(\tau_k,S_{\tau_k},v_{\tau_k})$ and
(ii) enters a fresh one-month 10\% OTM put with maturity
$T^{(k)}_{\mathrm{opt}}=\tau_k+\Delta_{\mathrm{opt}}$. \emph{The rolls
occur strictly before the natural expiry of the incumbent contract}:
concretely we take $\Delta<\Delta_{\mathrm{opt}}$ so that at each roll
$\tau_k$ the outgoing contract has residual time value
$T^{(k-1)}_{\mathrm{opt}}-\tau_k = \Delta_{\mathrm{opt}}-\Delta>0$. In the
simulations, $\Delta=21$ trading days and $\Delta_{\mathrm{opt}}=30$
calendar days, so the outgoing contract has roughly nine calendar days of
residual time value at every roll. Because the overlay state is
$(S_t,v_t)$ with $v_t>0$ and the residual time to expiry is bounded below
by $\Delta_{\mathrm{opt}}-\Delta>0$, the incumbent price
$P^{(k-1)}(\tau_k,S_{\tau_k},v_{\tau_k})$ is strictly positive almost
surely and the roll ratio below is bounded.

Let $P^{(k)}(t,s,v)$ denote the fixed-maturity price of the $k$-th
contract between $\tau_k$ and $\tau_{k+1}$. The rolled strategy wealth
$P^{\mathrm{roll}}_t$ evolves as $P^{(k)}(t,S_t,v_t)$ on
$(\tau_k,\tau_{k+1}]$, with a self-financing discontinuity at each
$\tau_{k+1}$ that exchanges the still-live incumbent contract for a fresh
one at market. In continuous time,
\begin{equation}
\frac{dP^{\mathrm{roll}}_t}{P^{\mathrm{roll}}_{t^-}}
=
\frac{dP^{(k)}_t}{P^{(k)}_{t^-}}\;\mathbf{1}_{(\tau_k,\tau_{k+1})}(t)
\;+\;
\underbrace{\left(\frac{P^{(k+1)}(\tau_{k+1},S_{\tau_{k+1}},v_{\tau_{k+1}})}
{P^{(k)}(\tau_{k+1},S_{\tau_{k+1}},v_{\tau_{k+1}})}-1\right)}_{
\displaystyle =: \eta_{k+1}}
\;d\mathbf{1}_{\{t\ge\tau_{k+1}\}}.
\label{eq:rolled_strategy}
\end{equation}
Between rolls, $P^{\mathrm{roll}}$ has exactly the dynamics
\eqref{eq:put_return} for a fixed-maturity contract. At each roll
$\tau_{k+1}$, the pre-expiry convention above guarantees
$P^{(k)}(\tau_{k+1},\cdot)>0$, so $\eta_{k+1}$ is a well-defined
finite return. Economically, $\eta_{k+1}$ captures the wealth-fraction
change when the portfolio substitutes the incumbent contract for a fresh
constant-maturity contract at prevailing market prices. All systematic
rolling costs --- transaction costs, bid--ask spread, and the
variance-risk-premium drag cited in Remark~\ref{rem:carry} --- are
absorbed into $\eta_{k+1}$ and into the between-roll drift $\mu_P$. An
equivalent implementation rebalances a fixed wealth fraction (or a fixed
premium budget) into the new contract and charges the residual cash flow
to the risk-free account; under bounded roll ratios the two
implementations coincide at the level of portfolio returns, so we retain
the notional form \eqref{eq:rolled_strategy} for consistency with the
wealth equation below.
Assumption~\ref{ass:overlay} applies to each $P^{(k)}$ individually. The theoretical
HJB analysis below is stated for the fixed-maturity object $P(t,s,v)$; the
rolled-strategy object $P^{\mathrm{roll}}_t$ inherits those dynamics with
the additional scheduled but state-dependent jump term in \eqref{eq:rolled_strategy}. It
can be used directly inside the CVaR problem without modification,
provided the roll dates $\{\tau_k\}$ are fixed in advance and the roll
ratios $\eta_{k+1}$ are bounded (as they are under the pre-expiry
convention). In the stylized simulations we measure $P(t,s,v)$ through a
Black--Scholes proxy calibrated to rolling realized volatility; this is
not consistent with the Heston--jump specification of
\eqref{eq:S}--\eqref{eq:v}, and the magnitudes of carry drag and
crash-state repricing are inherited from this proxy, so the reported CVaR
levels and hybrid-weight location should be read as illustrative rather
than invariant to the pricing model. A full Heston--jump pricing of each
$P^{(k)}$, as part of a follow-up empirical
calibration to a live option surface, is a natural extension.

\subsection{Trend signal and filtering bridge}

Define the exponentially weighted moving-average signal on log-returns by
\begin{equation}
M_t := \int_0^t e^{-\gamma(t-u)}\,dy_u,
\qquad \gamma>0.
\label{eq:Mdef}
\end{equation}
Because $dy_u$ is a log-return increment, $M_t$ is dimensionless and can be fed
directly into a position map $f$. The signal dynamics are
\begin{equation}
 dM_t = dy_t - \gamma M_t\,dt,
 \label{eq:dM}
\end{equation}
so a log-jump of size $\zeta$ moves the signal by exactly the same amount:
$\Delta M_t = \zeta$. This unit consistency is important. It means that the
reversal condition for trend can be stated in the same scale as the jump itself.

Take the directional exposure to be of the form
\begin{equation}
\pi_t = f(M_t),
\label{eq:pit}
\end{equation}
where $f$ is increasing and odd, for example $f(m)=\tanh(\beta m)$.

A complementary interpretation comes from modelling log-returns directly as
noisy observations of a latent trend. Fix a discretization step $\Delta>0$
and let $r_k:=y_{k\Delta}-y_{(k-1)\Delta}$ denote the log-return over the
$k$-th interval. Consider the local-level state-space model
\begin{align}
 r_k &= \mu_k + \varepsilon_k,
 &\varepsilon_k &\sim \mathcal{N}(0,\sigma_\varepsilon^2), \\
 \mu_k &= \mu_{k-1} + \omega_k,
 &\omega_k &\sim \mathcal{N}(0,\sigma_\omega^2),
\end{align}
in which $\mu_k$ is a latent drift rate and $\varepsilon_k,\omega_k$ are
independent Gaussian noises. The steady-state Kalman filter reduces to an
EWMA on returns with constant gain $\kappa^\star\in(0,1)$ determined by the
signal-to-noise ratio $\sigma_\omega^2/\sigma_\varepsilon^2$: the one-step
filtered estimate satisfies $\hat\mu_k=\kappa^\star r_k
+(1-\kappa^\star)\hat\mu_{k-1}$. Summing from the origin shows that
$\hat\mu_k$ is an exponentially weighted moving average of past returns,
which is the continuous-time limit of \eqref{eq:Mdef} up to scaling. This
provides a clean bridge between signal extraction and the control problem:
trend is a filtered estimate of persistent directional motion in returns
rather than an ad hoc trading rule, and the signal is defined on the same
scale (log-returns) on which the control action operates.

\begin{lemma}[Signal mechanics]
\label{lem:signal_mechanics}
Assume $f$ is odd and strictly increasing.
\begin{enumerate}[label=(\roman*),leftmargin=1.6em]
    \item If a log-jump $\zeta<0$ occurs at time $\tau$, then
    $M_\tau=M_{\tau^-}+\zeta$. The position changes sign on impact if and only if
    \begin{equation}
    \zeta < -M_{\tau^-}.
    \label{eq:jump_threshold}
    \end{equation}
    \item Suppose instead that after time $\tau$ the asset experiences a gradual
    deterministic drawdown with $dy_t=-c\,dt$ for $c>0$ and no further jumps.
    If $M_\tau>0$, then for $t\ge \tau$,
    \begin{equation}
    M_t = -\frac{c}{\gamma} + \Bigl(M_\tau+\frac{c}{\gamma}\Bigr)e^{-\gamma(t-\tau)}.
    \label{eq:Mdrawdown}
    \end{equation}
    The signal crosses zero at
    \begin{equation}
    t^* = \tau + \frac{1}{\gamma}\log\Bigl(1+\frac{\gamma M_\tau}{c}\Bigr),
    \label{eq:reversal_time}
    \end{equation}
    provided the drawdown lasts that long.
\end{enumerate}
\end{lemma}

\begin{proof}
Part (i) is immediate from \eqref{eq:dM}: a log-jump of size $\zeta$ enters both
$dy_t$ and $M_t$ one-for-one. Because $f$ is odd and strictly increasing,
$f(M_{\tau^-}+\zeta)<0$ if and only if $M_{\tau^-}+\zeta<0$.

For part (ii), under $dy_t=-c\,dt$ and no jumps, \eqref{eq:dM} becomes the ODE
$dM_t/dt = -c-\gamma M_t$. Solving it on $[\tau,\infty)$ gives
\eqref{eq:Mdrawdown}. Setting $M_t=0$ yields \eqref{eq:reversal_time}.
\end{proof}

\begin{figure}[tbp]
\centering
\includegraphics[width=0.96\textwidth]{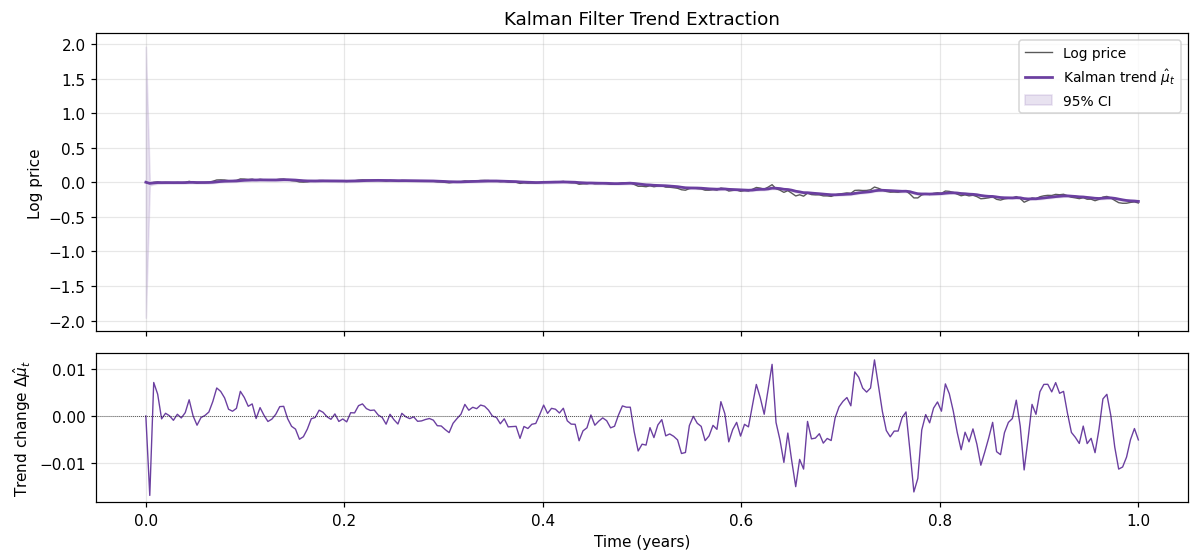}
\caption{Kalman-filter trend extraction on a representative simulated path from
our experiments. The top panel shows the log price together with the filtered
latent trend and a 95\% band. The bottom panel shows filtered trend increments.
At steady state the Kalman recursion is an EWMA, which is the signal used in the
control problem.}
\label{fig:kalman}
\end{figure}

\FloatBarrier
\section{Tail-risk objective}
\label{sec:cvar}

We work throughout with the log-loss convention
\begin{equation}
\label{eq:loss_convention}
L := -\log(X_T/X_0),
\end{equation}
i.e.\ terminal loss is measured on the continuously compounded scale. This
is equivalent to $L=-\log X_T$ when $X_0=1$ (which we assume without loss of
generality), and matches the scale on which the stylized Monte Carlo CVaRs
reported in \Cref{sec:mc} are computed. The log-loss convention is also
convenient because post-jump wealth is strictly positive under
Lemma~\ref{lem:positivity}, so $\log X_T$ is almost surely finite. For a
confidence level $\alpha\in(0,1)$, the Value-at-Risk and Conditional
Value-at-Risk are
\begin{equation}
\VaR_\alpha(L) := \inf\{\ell\in\R:\Prob(L\le \ell)\ge \alpha\},
\label{eq:VaR}
\end{equation}
and
\begin{equation}
\CVaR_\alpha(L) := \E\bigl[L\mid L\ge \VaR_\alpha(L)\bigr].
\label{eq:CVaR}
\end{equation}
The conditional-expectation form \eqref{eq:CVaR} requires the loss distribution
to be continuous at the $\alpha$-quantile. The Rockafellar--Uryasev
representation below is valid without this restriction and is the formulation
used throughout the paper.
The representation of \citet{Rockafellar2000} is central:
\begin{equation}
\CVaR_\alpha(L)
= \min_{\eta\in\R}
\left\{
\eta + \frac{1}{1-\alpha}\E\bigl[(L-\eta)^+\bigr]
\right\}.
\label{eq:RU}
\end{equation}
It converts a tail expectation into a stochastic program with an auxiliary scalar
$\eta$, which is exactly what makes dynamic programming feasible.

\begin{remark}[Pre-commitment CVaR]
\label{rem:precommitment}
Terminal CVaR is not in general a dynamically time-consistent risk measure:
an optimal policy for a terminal-$\CVaR$ problem at time $t=0$ need not
remain optimal for the corresponding tail problem at a later time
$t\in(0,T)$ \citep{Artzner1999,Rockafellar2000}. The framework in this
paper is therefore interpreted as a \emph{pre-commitment} problem: at
$t=0$ the investor commits to a feedback policy $(\pi_t,q_t)$ and an
auxiliary threshold $\eta$ by jointly minimizing the Rockafellar--Uryasev
expression in \eqref{eq:RU}; for any fixed $\eta$, the inner expectation
is a standard terminal-cost expected-value problem and admits a dynamic
programming representation, which we exploit in the HJB analysis below.
This is consistent with almost all of the applied CVaR control literature.
Fully time-consistent dynamic risk measures --- such as recursive
composition of one-period CVaR kernels --- are a separate construction and
are not the object of this paper.
\end{remark}

\subsection{Diagnostic decomposition of hedge quality}
\label{sec:diagnostics}

The CVaR objective is deliberately scalar: it ranks policies by the expected loss
inside the left tail. For interpretation and structuring, however, it is useful
to decompose the protection mechanism into observable hedge-quality dimensions.
Fix a baseline market portfolio with terminal loss $L^0$ and define the stress
set
\begin{equation}
\mathcal{A}_\alpha^0:=\{L^0\ge \VaR_\alpha(L^0)\}.
\label{eq:stressset}
\end{equation}
For a tradable hedge sleeve $i$ with horizon return $R_i$ and market return
$R_M$, define the downside and normal-state betas
\begin{equation}
\beta_i^-:=
\frac{\operatorname{Cov}(R_i,R_M\mid R_M\le q_\alpha^M)}
{\operatorname{Var}(R_M\mid R_M\le q_\alpha^M)},
\qquad
\beta_i^0:=
\frac{\operatorname{Cov}(R_i,R_M\mid R_M> q_\alpha^M)}
{\operatorname{Var}(R_M\mid R_M> q_\alpha^M)},
\label{eq:downside_beta}
\end{equation}
where $q_\alpha^M$ is a lower-tail market-return quantile. A simple empirical
conditional-convexity score is then
\begin{equation}
\mathsf{Conv}_{i,\alpha}:=-(\beta_i^- - \beta_i^0).
\label{eq:conv_score}
\end{equation}
For a protective hedge, a high value means that market beta becomes more negative
in bad states. This is an empirical co-movement proxy for conditional convexity;
in the option component of the model, the structural source is the jump repricing
$R_P(t,s,v;\zeta)$ and the curvature of $P$.

The other three dimensions can be measured on the same simulated or empirical
paths:
\begin{align}
\mathsf{Hit}_{i,\alpha}
&:=\Prob(R_i>0\mid \mathcal{A}_\alpha^0),
\label{eq:hit}\\
\mathsf{Carry}_{i,\alpha}
&:=-\E[R_i\mid (\mathcal{A}_\alpha^0)^c],
\label{eq:carrydiag}\\
\mathsf{Pers}_{i,h,d}
&:=\Prob\left(\sum_{j=1}^{h}R_{i,t+j}>0\,\middle|\,\mathcal{D}_{t,h,d}\right),
\label{eq:persistence}
\end{align}
where $\mathcal{D}_{t,h,d}$ denotes a drawdown episode of duration $h$ and depth
at least $d$. Thus reliability is the tail-event hit ratio, carry is the
non-stress drag of the sleeve, and persistence is the probability that the hedge
continues to help over a multi-period drawdown rather than only on the initial
shock.

\begin{table}[tbp]
\centering
\caption{Four hedge-quality dimensions and their role in the present model.}
\label{tab:diagnostics}
\begin{threeparttable}
\begin{adjustbox}{max width=0.99\textwidth}
\begin{tabularx}{\textwidth}{p{0.22\textwidth}X X}
\toprule
Dimension & Diagnostic object & Model counterpart \\
\midrule
Conditional convexity & Downside co-movement asymmetry, measured for example by $\mathsf{Conv}_{i,\alpha}$ & Put jump repricing $R_P(t,s,v;\zeta)$, option elasticity $\delta_P$, and curvature of $P$ \\
Reliability & Tail hit ratio $\mathsf{Hit}_{i,\alpha}=\Prob(R_i>0\mid \mathcal{A}_\alpha^0)$ & Whether the sleeve pays in the same states that determine CVaR \\
Carry & Non-stress expected drag $\mathsf{Carry}_{i,\alpha}$ & Physical overlay drift $\mu_P$ and the premium drag of rolled convex insurance \\
Persistence & Multi-period drawdown protection $\mathsf{Pers}_{i,h,d}$ & Signal half-life $\gamma^{-1}$, reversal time $t^*-\tau$, and trend exposure $f(M_t)$ \\
\bottomrule
\end{tabularx}
\end{adjustbox}
\begin{tablenotes}[flushleft]
\footnotesize
\item These diagnostics do not replace CVaR. They explain which mechanism reduces
CVaR and provide economically interpretable constraints for implementation.
\end{tablenotes}
\end{threeparttable}
\end{table}

This decomposition leads to two implementable variants of the baseline problem.
The first is a constrained version,
\begin{equation}
\begin{aligned}
\min_{\theta}\quad & \CVaR_\alpha(L(\theta)) \\
\text{s.t.}\quad
& \mathsf{Conv}_{\theta,\alpha}\ge \underline C,
\qquad
\mathsf{Hit}_{\theta,\alpha}\ge \underline H, \\
& \mathsf{Carry}_{\theta,\alpha}\le \overline K,
\qquad
\mathsf{Pers}_{\theta,h,d}\ge \underline P,
\end{aligned}
\label{eq:constrained_diagnostic}
\end{equation}
where $\theta$ parameterizes the hedge policy or sleeve weights. The second is a
penalized multi-objective version,
\begin{equation}
\min_{\theta}\;
\CVaR_\alpha(L(\theta))
+\lambda_K\mathsf{Carry}_{\theta,\alpha}
+\lambda_F\bigl(1-\mathsf{Hit}_{\theta,\alpha}\bigr)
+\lambda_D\bigl(1-\mathsf{Pers}_{\theta,h,d}\bigr)
-\lambda_C\mathsf{Conv}_{\theta,\alpha}.
\label{eq:penalized_diagnostic}
\end{equation}
The multipliers have a direct economic interpretation. A crash-sensitive mandate
places high weight on conditional convexity and hit ratio; a cost-sensitive
mandate puts more weight on carry; and a mandate concerned with long bear markets
places more weight on persistence. In numerical work, the non-smooth indicators
in \eqref{eq:hit} and \eqref{eq:persistence} can be replaced by smooth logistic
approximations, or handled directly by simulation-based constrained optimization.

\section{Dynamic control problem}
\label{sec:control}

\subsection{Wealth dynamics and Markov state}

Let $X_t$ denote wealth. Given directional exposure $\pi_t$ and convex-overlay
exposure $q_t$, wealth evolves as
\begin{equation}
\frac{dX_t}{X_{t^-}} = \pi_t\frac{dS_t}{S_{t^-}} + q_t\frac{dP_t}{P_{t^-}}.
\label{eq:Xbasic}
\end{equation}
Using \eqref{eq:S} and \eqref{eq:put_return},
\begin{align}
\frac{dX_t}{X_{t^-}}
&= \bigl[\pi_t\mu + q_t\mu_P(t,S_t,v_t)\bigr]dt \notag\\
&\quad + \bigl[\pi_t + q_t\delta_P(t,S_t,v_t)\bigr]\sqrt{v_t}\,dW_t^S
+ q_t\nu_P(t,S_t,v_t)\xi\sqrt{v_t}\,dZ_t \notag\\
&\quad + \bigl[\pi_t(e^{\zeta}-1)+q_tR_P(t,S_{t^-},v_t;\zeta)\bigr]dN_t.
\label{eq:Xexpanded}
\end{align}
The immediate post-jump wealth map is therefore
\begin{equation}
X_t = X_{t^-}\Bigl(1+\pi_t(e^{\zeta}-1)+q_tR_P(t,S_{t^-},v_t;\zeta)\Bigr).
\label{eq:jumpwealth}
\end{equation}

\begin{lemma}[Strict wealth positivity under two-sided directional exposure]
\label{lem:positivity}
Assume Assumption~\ref{ass:admissible} and the overlay bound $R_P(t,s,v;\zeta)\ge -1$
for all admissible $(t,s,v,\zeta)$. Then the post-jump wealth multiplier in
\eqref{eq:jumpwealth} satisfies, for every admissible $(\pi_t,q_t)$ and every
$\zeta\in[\underline\zeta,\overline\zeta]$,
\[
1+\pi_t(e^{\zeta}-1)+q_tR_P(t,S_{t^-},v_t;\zeta)
\;\ge\;
\min\!\Bigl\{
1-\overline\pi-\overline q,\;
1+\underline\pi(e^{\overline\zeta}-1)-\overline q
\Bigr\}
\;>\;0,
\]
where the second inequality uses the leverage condition
\eqref{eq:leverage_cond}. Consequently, if $X_0>0$, the wealth process
remains strictly positive almost surely on $[0,T]$.
\end{lemma}

\begin{proof}
Between jump times, $X$ is the stochastic exponential of a continuous
semimartingale and is therefore strictly positive as long as it starts
positive. At a jump time, the multiplier
$m(\zeta):=1+\pi_t(e^{\zeta}-1)+q_tR_P(t,S_{t^-},v_t;\zeta)$ is the sum of
three terms. The worst case for a long position ($\pi_t\ge 0$) is
$\zeta=\underline\zeta$, which gives the bound $1-\overline\pi-\overline q$
because $e^{\underline\zeta}-1\ge -1$ and $R_P\ge -1$. The worst case for a
short position ($\pi_t<0$) is $\zeta=\overline\zeta$, which gives the bound
$1+\underline\pi(e^{\overline\zeta}-1)-\overline q$; this is strictly positive
by \eqref{eq:leverage_cond}. Both bounds depend only on fixed admissible
parameters, so the multiplier is bounded away from zero uniformly in $\zeta$
and in $(\pi_t,q_t)\in\mathcal{K}$, which preserves positivity at every jump.
\end{proof}

\begin{remark}[Economic interpretation of the leverage condition]
The two inequalities in \eqref{eq:leverage_cond} handle the two worst-case
jump directions separately. The first, $\overline\pi+\overline q<1$,
restricts how much simultaneous long directional and overlay exposure can
be admitted: in the worst-case downside jump
($e^{\underline\zeta}-1\ge -1$ combined with $R_P\ge -1$) the post-jump
multiplier collapses to $1-\overline\pi-\overline q$, which must be strictly
positive. The second inequality is the symmetric constraint for a short
directional position hit by the worst upside jump: it couples the lower
directional bound $\underline\pi<0$, the upside truncation $\overline\zeta$,
and the overlay upper bound $\overline q$. Both inequalities are non-binding
in the stylized calibration of \Cref{tab:params}: the annualized jump
intensity and mean log-jump are negative, and the upside truncation
$\overline\zeta$ sits several standard deviations above the mean, while the
numerical experiments take $\overline\pi+\overline q$ well below $1$. If the
modeller wishes to permit stronger directional or overlay positions, the
bounds must be reconciled through \eqref{eq:leverage_cond}, e.g.\ by
lowering $\overline q$ (reducing overlay concentration) or by tightening
$\overline\pi$ or $\underline\pi$. The alternative --- leaving jumps
unbounded and admitting wealth default --- requires a different admissible
class than the one used here.
\end{remark}

The natural Markov state is now explicit:
\begin{equation}
U_t := (X_t,S_t,v_t,M_t) \in (0,\infty)^3\times\R.
\end{equation}
Spot must be included separately from wealth because the overlay return depends
on the current asset price even when the investor is hedged. This is the main
state-variable correction relative to simplified treatments that try to infer the
underlying level from wealth alone.

\subsection{HJB characterization}

Fix the auxiliary parameter $\eta\in\R$ from \eqref{eq:RU}. For
$u=(x,s,v,m)$ define
\begin{equation}
V(t,u;\eta)
:= \inf_{(\pi,q)\in\mathcal{A}}
\left\{
\eta + \frac{1}{1-\alpha}\E\bigl[(-\log X_T-\eta)^+\mid U_t=u\bigr]
\right\},
\label{eq:Vdef}
\end{equation}
where $\mathcal{A}$ is the admissible class induced by Assumption~\ref{ass:admissible}
and $X_0=1$ so that $L=-\log X_T$ as in \eqref{eq:loss_convention}.
The full CVaR is recovered by the outer minimization
$\CVaR_\alpha(L)=\min_{\eta\in\R}V(t,U_t;\eta)$, so the problem has a
two-level structure: the HJB below solves the inner policy optimization for each
fixed $\eta$, and $\eta$ is chosen in an outer step.

To write the generator compactly, define the drift vector
\begin{equation}
b^{\pi,q}(t,u):=
\begin{pmatrix}
 x\bigl(\pi\mu + q\mu_P(t,s,v)\bigr) \\
 \mu s \\
 \kappa(\theta-v) \\
 \mu - \tfrac12 v - \gamma m
\end{pmatrix}
\label{eq:driftvec}
\end{equation}
and the diffusion loading matrix with respect to the correlated Brownian pair
$(W^S,Z)$,
\begin{equation}
\Sigma^{\pi,q}(t,u):=
\sqrt{v}
\begin{pmatrix}
 x(\pi+q\delta_P(t,s,v)) & xq\nu_P(t,s,v)\xi \\
 s & 0 \\
 0 & \xi \\
 1 & 0
\end{pmatrix},
\qquad
C:=\begin{pmatrix}1&\rho\\ \rho&1\end{pmatrix}.
\label{eq:sigma}
\end{equation}
The local covariance matrix is $a^{\pi,q}(t,u):=\Sigma^{\pi,q}(t,u)C\Sigma^{\pi,q}(t,u)^\top$.
Finally, define the jump map
\begin{equation}
\Gamma^{\pi,q}(t,u;\zeta)
:= \Bigl(x\bigl(1+\pi(e^{\zeta}-1)+qR_P(t,s,v;\zeta)\bigr),\; se^{\zeta},\; v,\; m+\zeta\Bigr).
\label{eq:jumpmap}
\end{equation}

\begin{proposition}[HJB equation in viscosity form]
\label{prop:HJB}
Assume Assumptions~\ref{ass:admissible} and~\ref{ass:overlay}. Suppose the coefficients of the state
process satisfy the usual local Lipschitz and linear-growth conditions and that
$V$ in \eqref{eq:Vdef} is finite and continuous. Then $V$ is a viscosity
solution of
\begin{equation}
\partial_t V + \inf_{(\pi,q)\in\mathcal{K}}
\Bigl\{
 b^{\pi,q}(t,u)\cdot \nabla_u V
 + \tfrac12\operatorname{tr}\bigl(a^{\pi,q}(t,u)D_u^2V\bigr)
 + \lambda\E_{\zeta}\bigl[V(t,\Gamma^{\pi,q}(t,u;\zeta);\eta)-V(t,u;\eta)\bigr]
\Bigr\}=0,
\label{eq:HJB}
\end{equation}
with terminal condition
\begin{equation}
V(T,u;\eta)=\eta+\frac{1}{1-\alpha}(-\log x-\eta)^+.
\label{eq:terminal}
\end{equation}
Under a comparison principle for the associated integro-PDE, this viscosity
solution is unique.
\end{proposition}

\begin{proof}[Proof sketch]
Under the stated regularity conditions, $U_t$ is a controlled Markov jump
diffusion. The dynamic-programming principle for controlled Markov processes
then yields \eqref{eq:HJB} with the generator built from the drift,
diffusion covariance, and jump map above; see \citet{FlemingSoner2006},
\citet{Pham2009}, and \citet{Oksendal2005}. The terminal cost is kinked because
of the positive-part function in \eqref{eq:RU}, so viscosity rather than
classical solutions are the natural solution concept. A comparison principle for
the resulting integro-PDE is a standard additional assumption rather than a consequence of bounded jumps alone.
Under compact controls, continuous coefficients with the growth bounds stated
above, finite-activity bounded jumps, and an admissible polynomial-growth class
for the value function, the required comparison result follows from the
nonlocal maximum principle for semicontinuous functions of
\citet{JakobsenKarlsen2006} and the Jensen--Ishii lemma for
integro-differential equations of \citet{BarlesImbert2008}. Bounded jump support
is used here to control the nonlocal term and finite moments of the jump
measure; the remaining continuity and growth hypotheses are part of the
standard viscosity-solution comparison framework for controlled jump diffusions.
\end{proof}

\begin{remark}[Why this accounting matters]
\label{rem:accounting}
The option overlay affects the HJB through \emph{all} three channels: drift
(carry), diffusion exposure, and jump repricing. This is exactly what a serious
tail-risk-management model must capture. A terminal-payoff expression of the form
$X_T=S_T+\theta(K-S_T)^+$ is useful for intuition, but it is not a self-financing
wealth equation and cannot by itself support dynamic optimization.
\end{remark}

\begin{remark}[Where the four diagnostics sit in the HJB]
\label{rem:diagnostics_hjb}
The HJB still optimizes a single CVaR objective. The diagnostic quantities in
\Cref{sec:diagnostics} are not additional state variables unless the modeller
chooses to impose them dynamically. They are lower-dimensional summaries of the
same simulated or controlled return distribution: jump repricing maps into
conditional convexity and hit ratio, the drift term maps into carry, and the
signal dynamics map into persistence. This is why the diagnostic layer improves
interpretability without changing the core state equation.
\end{remark}

\begin{proposition}[Classical verification for smooth candidates]
\label{prop:verification}
Let $E := (0,\infty)^3 \times \R$. For any smooth test function $\phi$, define the controlled integro-differential operator
\begin{align*}
(\mathcal{L}^{\pi,q}\phi)(t,u)
&:=
b_{\pi,q}(t,u)\cdot \nabla_u \phi(t,u)
+
\tfrac{1}{2}\operatorname{tr}\!\big(a_{\pi,q}(t,u) D_u^2 \phi(t,u)\big) \\
&\quad
+
\lambda \E_\zeta\!\left[
\phi\bigl(t,\Gamma_{\pi,q}(t,u;\zeta)\bigr) - \phi(t,u)
\right].
\end{align*}

Fix $\eta \in \R$. Suppose there exists
\[
W(\cdot,\cdot;\eta) \in C^{1,2}([0,T)\times E)\cap C([0,T]\times E)
\]
with polynomial growth such that
\[
\partial_t W(t,u;\eta)
+
\inf_{(\pi,q)\in\mathcal{K}}
(\mathcal{L}^{\pi,q}W)(t,u;\eta)
= 0
\]
for $(t,u)\in [0,T)\times E$, together with terminal condition
\[
W(T,u;\eta)
=
\eta + \frac{1}{1-\alpha}(-\log x-\eta)_+,
\qquad
u=(x,s,v,m).
\]
Assume also that the infimum is attained by a measurable selector
\[
(\pi^*,q^*) = (\pi^*,q^*)(t,u),
\]
that the associated controlled state process is well posed, and that the stopped local martingales generated by It\^o's formula are uniformly integrable. Then
\[
W(t,u;\eta)=V(t,u;\eta),
\]
and the feedback control
\[
(\pi_t^*,q_t^*)=(\pi^*,q^*)(t,U_t)
\]
is optimal for the inner problem.
\end{proposition}

\begin{proof}
Fix any admissible control $(\pi,q)$. Apply It\^o's formula for jump diffusions to $W(t,U_t;\eta)$, stopped at a localizing sequence. Since
\[
\partial_t W + \mathcal{L}^{\pi,q}W \ge 0
\]
by the HJB inequality, taking expectations yields
\[
W(t,u;\eta)
\le
\eta + \frac{1}{1-\alpha}
\E\!\left[(-\log X_T-\eta)_+ \,\middle|\, U_t=u\right].
\]
Taking the infimum over admissible controls gives $W \le V$.

For the feedback selector $(\pi^*,q^*)$, the HJB holds with equality. The
finite-variation term in It\^o's formula is then zero, and uniform integrability
of the stopped martingales gives
\[
W(t,u;\eta)
=
\eta + \frac{1}{1-\alpha}
\E^{\pi^*,q^*}\!\left[(-\log X_T-\eta)_+ \,\middle|\, U_t=u\right].
\]
Since $V$ is the infimum over all admissible controls, the right-hand side is
at least $V(t,u;\eta)$. Combined with $W\le V$ from the previous paragraph,
this yields $W=V$, and the feedback control is optimal.
\end{proof}

\section{Comparative statics and hybrid demand}
\label{sec:comparative}

\subsection{Why puts win on impact}

\begin{proposition}[Instantaneous crash protection]
\label{prop:jumpprotect}
Fix a pre-jump state $u_-=(x,s,v,m)$ and controls $(\pi,q)$. If a negative
log-jump $\zeta<0$ arrives immediately, post-jump wealth is given by
\eqref{eq:jumpwealth}, and its directional derivative with respect to convex
exposure is
\begin{equation}
\frac{\partial X_+}{\partial q} = xR_P(t,s,v;\zeta).
\label{eq:dXdqjump}
\end{equation}
Whenever the overlay reprices upward in the crash state,
$R_P(t,s,v;\zeta)>0$, increasing $q$ raises post-jump wealth and lowers the
corresponding one-step loss. The effect strengthens as the crash deepens if the
overlay is sufficiently convex in $s$.
\end{proposition}

\begin{proof}
Differentiate \eqref{eq:jumpwealth} with respect to $q$. The sign claim is
immediate. For convex overlays such as puts, deeper negative jumps raise the
mark-to-market revaluation term $R_P$ once the contract moves toward or into the
money.
\end{proof}

The derivative is trivial, but it sets up the key comparison with the trend
channel.
\begin{corollary}[Trend reversal threshold]
\label{cor:threshold}
Under the odd increasing map $\pi_t=f(M_t)$, there is a whole region of crash
sizes in which the trend channel remains long on impact while the put already
helps. Specifically, if $M_{t^-}>0$ and
\begin{equation}
-M_{t^-}<\zeta<0,
\end{equation}
then the post-jump trend position remains positive, whereas the convex overlay
still has the potential to reprice upward through $R_P(t,s,v;\zeta)$.
\end{corollary}

The economics are immediate. Puts respond state by state. Trend responds through
a signal, and the signal must first move through zero. That is why trend can be
late in violent reversals even when it is effective in drawn-out crises.

\begin{remark}[Impact versus regime]
\label{rem:impact_vs_regime}
Proposition~\ref{prop:jumpprotect} compares the two channels \emph{at one isolated jump
arrival}, conditional on the pre-jump state. That is a genuinely different
object from a one-year \emph{regime} characterized by repeated jumps, such
as the flash-crash calibration used in the numerical experiments of
\Cref{sec:mc}. A regime with intensity $\lambda$ experiences an expected
$\lambda T$ jumps over the horizon; between jumps the trend signal evolves
continuously, and under Lemma~\ref{lem:signal_mechanics} it drifts toward
negative values after any sustained period of negative returns. Over such
a regime, trend can therefore be already defensive by the time later jumps
arrive, even though on the very first jump Proposition~\ref{prop:jumpprotect} implies
that convex insurance is the instantaneous dominant channel. Thus the
``puts win on impact'' statement and the numerical finding that trend does
a large share of the CVaR reduction in a year-long flash-crash regime are
not in tension: they describe two different conditional objects. The
interior hybrid optima documented in \Cref{sec:mc} arise because both
mechanisms contribute: trend handles the between-jump drift, convex
insurance handles the impact revaluation, and the CVaR-minimizing
combination balances the two.
\end{remark}

\subsection{Why trend wins in persistent drawdowns}

Lemma~\ref{lem:signal_mechanics} gives the complementary result. Under a gradual
deterministic selloff, the signal follows \eqref{eq:Mdrawdown} and crosses zero
after the finite delay in \eqref{eq:reversal_time}. The trend overlay therefore
needs time, but once the drawdown horizon exceeds $t^*-\tau$, the same signal
mechanism that hurts it in jump states becomes an advantage: exposure turns
progressively defensive without requiring fresh premium outlays. That is the
precise sense in which trend is a hedge against \emph{persistent} rather than
\emph{instantaneous} tail events. Remark~\ref{rem:impact_vs_regime} extends
this observation to crash-prone regimes: a regime of repeated jumps still
admits trend protection provided the jumps are not so closely packed that
the signal has no time to react between them.

\subsection{Sufficient conditions for an interior hybrid optimum}

\paragraph{From free control to trend-following structure.}
Proposition~\ref{prop:HJB} characterizes the CVaR-optimal policy when the
directional control $\pi_t$ is free within the compact set $\mathcal{K}$.
Under that formulation the model is more general than a trend-following
model: any progressively measurable, compactly bounded directional policy
is admissible, and the trend signal $M_t$ enters only as a state variable
that the optimizer may or may not use. To represent a trend-following
overlay specifically, the directional control must be restricted to depend
on the signal in a monotonic way --- that is, to take the form $\pi_t =
g_\theta(t,X_t,S_t,v_t,M_t)$ with $g_\theta$ increasing in $M_t$. The
simplest nontrivial instance is $\pi_t=f(M_t)$ with $f$ odd and strictly
increasing, as introduced in \eqref{eq:pit}. The comparative statics of
\Cref{sec:comparative} and the simulations of \Cref{sec:mc} are stated for
this restricted class, not for the unrestricted HJB problem. The
framework thus nests two objects: the unrestricted CVaR-optimal control,
which establishes a reference ceiling, and the trend-following subclass,
which is the structurally interpretable policy actually implemented in
practice. The reduced two-parameter family below is the further
simplification used throughout the numerical work.

The full HJB problem is high dimensional. To isolate the economic content of the
hybrid demand result, it is useful to consider a reduced two-parameter policy
class:
\begin{equation}
\pi_t = a f(M_t), \qquad q_t=b,
\qquad (a,b)\in[0,\overline a]\times[0,\overline b].
\label{eq:reducedclass}
\end{equation}
Let $L^{a,b}$ be the resulting terminal loss and define
\begin{equation}
J(a,b):=\CVaR_\alpha(L^{a,b}).
\end{equation}
This reduced problem is exactly the object visualized in our experimental grid
search over the hybrid put weight.

\begin{proposition}[Sufficient conditions for an interior hybrid optimum in the reduced class]
\label{prop:interior}
Suppose $J$ is continuously differentiable and strictly convex on the rectangle
$[0,\overline a]\times[0,\overline b]$. Assume moreover that for every fixed
$b\in[0,\overline b]$,
\begin{equation}
\partial_a J(0,b)<0<\partial_a J(\overline a,b),
\label{eq:agrad}
\end{equation}
and for every fixed $a\in[0,\overline a]$,
\begin{equation}
\partial_b J(a,0)<0<\partial_b J(a,\overline b).
\label{eq:bgrad}
\end{equation}
Then $J$ admits a unique minimizer $(a^*,b^*)$ satisfying
$0<a^*<\overline a$ and $0<b^*<\overline b$.
\end{proposition}

\begin{proof}
Because $J$ is continuous on the compact rectangle
$[0,\overline a]\times[0,\overline b]$, a minimizer exists by the extreme-value
theorem. Strict convexity implies that $J$ has at most one minimizer. The boundary
sign conditions \eqref{eq:agrad}--\eqref{eq:bgrad} rule out every boundary
point: the objective is decreasing when one enters the rectangle from the left
or bottom edge and increasing before reaching the right or top edge. Hence the
unique minimizer cannot lie on the boundary and must be interior.
\end{proof}

\begin{remark}[Status of the result]
\label{rem:interior_status}
Proposition~\ref{prop:interior} is a formalization of the intuition that a well-behaved
CVaR objective with non-trivial boundary slopes admits an interior
hybrid solution; it is not a general theorem that hybrid configurations
dominate pure put or pure trend allocations. Both the strict-convexity
hypothesis and the two boundary-slope conditions are substantive and do not
hold automatically: the first requires the CVaR objective to be strictly
convex in $(a,b)$, which generally fails for CVaR on nonlinear-in-weight
losses; the second asks that both sleeves strictly improve the tail at
zero allocation and that neither dominates at maximal allocation. The
proposition should therefore be read as giving \emph{sufficient} conditions
under which the numerically observed interior optima of \Cref{fig:optimalw}
are consistent with a well-posed optimization problem, rather than as a
claim that interior hybrids are generically optimal.
\end{remark}

\begin{remark}[Economic interpretation]
The two sign conditions have clear content. The first says that allowing a
trend-sensitive directional overlay strictly improves tail risk relative to zero
trend exposure. The second says that adding some convex insurance strictly
improves tail risk relative to no insurance. When both statements are true and
the reduced objective is well behaved, the optimum is genuinely hybrid rather
than a corner solution.
\end{remark}

\begin{remark}[On the strict-convexity assumption]
\label{rem:convexity}
CVaR is convex in portfolio weights when the loss is linear in the weights, but
$J(a,b)$ here involves a nonlinear dependence through the trend signal and the
option repricing. Strict convexity therefore does not hold automatically; it is
a regularity condition on the particular parameter combination and jump-size
distribution. In the simulation section, the grid-search CVaR curves in
\cref{fig:optimalw} are visually convex and admit clear interior minima, which is
consistent with the assumption holding in the stylized calibration. More
generally, strict convexity can fail if the two hedge channels are nearly
redundant, in which case the optimum becomes a flat ridge rather than a unique
point.
\end{remark}

\begin{proposition}[Local characterization of a hybrid optimum in the reduced class]
\label{prop:local_hybrid}
Suppose
\[
J(a,b)=\CVaR_\alpha(L_{a,b})
\]
is twice continuously differentiable on an open neighborhood of $(a^*,b^*)\in (0,\bar a)\times(0,\bar b)$. If
\[
\nabla J(a^*,b^*) = 0
\]
and the Hessian
\[
H_J(a^*,b^*)
=
\begin{pmatrix}
\partial_{aa}J(a^*,b^*) & \partial_{ab}J(a^*,b^*) \\
\partial_{ba}J(a^*,b^*) & \partial_{bb}J(a^*,b^*)
\end{pmatrix}
\]
is positive definite, equivalently
\[
\partial_{aa}J(a^*,b^*) > 0,
\qquad
\det H_J(a^*,b^*) > 0,
\]
then $(a^*,b^*)$ is a strict local minimizer of $J$. If, in addition, $J$ is convex on
$[0,\bar a]\times[0,\bar b]$, then this local minimizer is the unique global minimizer on that rectangle.
\end{proposition}

\begin{proof}
The first statement is the standard second-order sufficient condition for a strict local minimum. If $J$ is convex on the full rectangle, any local minimizer is automatically global. Positive definiteness of the Hessian yields local uniqueness, and convexity upgrades this to uniqueness on the admissible rectangle.
\end{proof}

\begin{remark}[Why this is sharper than global strict convexity]
The point of the proposition is to separate two logically distinct issues. Local stability of the hybrid solution is governed by the Hessian at the candidate optimum, whereas global uniqueness on the admissible set requires an additional convexity argument. In applications, the first question is often the economically relevant one: whether the numerically identified hybrid point is genuinely interior and locally well pinned down.
\end{remark}

\section{CVaR policy gradients}
\label{sec:gradient}

Let $L(\bm\theta):=-\log X_T(\bm\theta)$ be the terminal log-loss generated
by a parameterized policy $\bm\theta\mapsto(\pi_{\bm\theta},q_{\bm\theta})$,
under the loss convention \eqref{eq:loss_convention}.

\begin{theorem}[CVaR policy-gradient identity]
\label{thm:pg}
Assume $L(\bm\theta)$ is almost surely continuously differentiable in
$\bm\theta$, that its gradient is dominated by an integrable envelope in a
neighborhood of the parameter of interest, and that the distribution of
$L(\bm\theta)$ is continuous at the $\alpha$-quantile. Then
\begin{equation}
\nabla_{\bm\theta}\CVaR_\alpha\bigl(L(\bm\theta)\bigr)
=
\frac{1}{1-\alpha}
\E\Bigl[\nabla_{\bm\theta}L(\bm\theta)\,\mathbf{1}_{\{L(\bm\theta)\ge
\VaR_\alpha(L(\bm\theta))\}}\Bigr].
\label{eq:pg}
\end{equation}
\end{theorem}

\begin{proof}
Start from the Rockafellar--Uryasev representation \eqref{eq:RU}. Under the
dominating envelope, differentiation may be interchanged with expectation.
Continuity of the loss distribution at the $\alpha$-quantile ensures that the
positive-part term is a.s. differentiable at $\eta^*=\VaR_\alpha(L(\bm\theta))$
and that the boundary term in the interchange vanishes; see
\citet[Thm.~3.1]{HongLiu2009} for a direct proof of \eqref{eq:pg} under these
assumptions. An equivalent derivation applies the envelope theorem to the
outer minimization in \eqref{eq:RU} after absorbing the optimal $\eta^*$.
\end{proof}

\begin{corollary}[Tail-gradient interpretation of hybrid demand]
\label{cor:tailgrad}
In the reduced class
\[
\pi_t = a f(M_t), \qquad q_t = b,
\]
let
\[
J(a,b):=\CVaR_\alpha(L_{a,b}).
\]
Under the regularity conditions of Theorem~\ref{thm:pg},
\[
\partial_a J(a,b)
=
\frac{1}{1-\alpha}
\E\!\left[
\partial_a L_{a,b}\,
\mathbf{1}\!\left\{L_{a,b}\ge \VaR_\alpha(L_{a,b})\right\}
\right],
\]
and
\[
\partial_b J(a,b)
=
\frac{1}{1-\alpha}
\E\!\left[
\partial_b L_{a,b}\,
\mathbf{1}\!\left\{L_{a,b}\ge \VaR_\alpha(L_{a,b})\right\}
\right].
\]
Hence the first-order condition for an interior hybrid optimum is equivalent to
\[
\E\!\left[
\partial_a L_{a^*,b^*}\,
\mathbf{1}\!\left\{L_{a^*,b^*}\ge \VaR_\alpha(L_{a^*,b^*})\right\}
\right]
=0,
\]
\[
\E\!\left[
\partial_b L_{a^*,b^*}\,
\mathbf{1}\!\left\{L_{a^*,b^*}\ge \VaR_\alpha(L_{a^*,b^*})\right\}
\right]
=0.
\]
In words: at the interior hybrid optimum, neither an incremental increase in trend intensity nor an incremental increase in convex insurance improves the $\alpha$-tail on average.
\end{corollary}

\begin{proof}
Apply Theorem~\ref{thm:pg} with parameter $\bm\theta=(a,b)$.
\end{proof}

The sample analogue should be written carefully. If $\hat q_\alpha$ is the
empirical $\alpha$-quantile and
$N_\alpha:=\sum_{i=1}^N \mathbf{1}_{\{L^{(i)}\ge \hat q_\alpha\}}$, then a
natural Monte Carlo estimator is
\begin{equation}
\widehat{\nabla_{\bm\theta}\CVaR_\alpha}
=
\frac{1}{N_\alpha}
\sum_{i=1}^N
\nabla_{\bm\theta}L^{(i)}\,\mathbf{1}_{\{L^{(i)}\ge \hat q_\alpha\}}.
\label{eq:pg_est}
\end{equation}
Because this is the tail average of the loss gradient, it is the finite-sample object
that corresponds directly to the conditional-expectation definition of CVaR.
Under standard regularity conditions, it is a consistent estimator of the policy
gradient; see again \citet{HongLiu2009}. The alternative normalization
$\tfrac{1}{(1-\alpha)N}\sum_{i=1}^N\nabla_{\bm\theta}L^{(i)}
\mathbf{1}_{\{L^{(i)}\ge\hat q_\alpha\}}$ that drops out of
\eqref{eq:RU} directly is asymptotically equivalent to \eqref{eq:pg_est}
because $N_\alpha/N\to 1-\alpha$ a.s.\ under distributional continuity at
the quantile; the difference is $O_p(N^{-1/2})$ in the tail-count ratio and
does not affect consistency. In practice we also use common random numbers
across the $w$-grid in \Cref{fig:gradient} so that the tail-gradient curve
is driven by policy changes rather than by simulation noise between grid
points. In constrained problems such as $w\in[0,1]$, projected gradient
descent is immediate.

\section{Stylized Monte Carlo evidence}
\label{sec:mc}

This section reports numerical experiments that generate fast, stylized evidence
for the theory. The goal is not full empirical calibration. It is to check
whether a computational implementation reproduces the state ranking implied by
the analytical results above.

\subsection{Design}

The experiments simulate one-year paths at daily frequency under the model in
\Cref{sec:setup} using $2{,}000$ Monte Carlo paths in the baseline and three
stress scenarios. For speed, the put overlay is approximated by a monthly rolled
10\% OTM Black--Scholes put whose volatility input is a rolling realized-volatility
proxy. Trend is implemented as an equal-weight combination of 1-, 3-, and
12-month EWMA signals on simulated log-returns. The hybrid overlay is a convex
combination of the put and trend sleeves. Because the reduced policy class
\eqref{eq:reducedclass} uses a constant hybrid weight $w$ along each path, the
pre-commitment issue of Remark~\ref{rem:precommitment} is less severe in this
section than in a fully reoptimized dynamic policy problem. It does not disappear
conceptually: the simulated objective remains terminal CVaR, and the weight is
chosen ex ante. The numerical exercise should therefore be interpreted as a
pre-commitment allocation experiment over a transparent two-sleeve strategy
class.

The grid search includes the pure sleeves as endpoints, so an optimized hybrid
can mechanically improve on the endpoints in-sample when the CVaR curve is
convex. For that reason the evidence below reports both the fixed equal-weight
Hybrid in \Cref{tab:baseline} and the optimized grid weights in
\Cref{tab:scenarios}. All grid points use common random numbers, which reduces
Monte Carlo noise in cross-weight comparisons. In a production empirical study,
the optimized weight should additionally be selected on a training period or
simulation batch and evaluated on an independent holdout batch.

\begin{table}[tbp]
\centering
\caption{Stylized simulation inputs used in the experiments.}
\label{tab:params}
\begin{threeparttable}
\begin{adjustbox}{max width=0.97\textwidth}
\begin{tabular}{llll}
\toprule
Block & Quantity & Symbol & Value \\
\midrule
Market & Expected return & $\mu$ & $0.07$ \\
Market & Initial / long-run variance & $v_0,\theta$ & $0.04,\;0.04$ \\
Market & Mean reversion & $\kappa$ & $3.0$ \\
Market & Vol-of-vol & $\xi$ & $0.40$ \\
Market & Leverage correlation & $\rho$ & $-0.70$ \\
Jumps & Intensity & $\lambda$ & $3.0$ per year \\
Jumps & Mean log-jump & $\mu_\zeta$ & $-0.05$ \\
Jumps & Log-jump volatility & $\sigma_\zeta$ & $0.07$ \\
Trend & Signal half-lives & --- & 1m / 3m / 12m \\
Put overlay & Moneyness / roll frequency & $K/S_0$ & $0.90$, monthly \\
Simulation & Horizon / step / paths & --- & 1 year / 252 days / $2{,}000$ \\
\bottomrule
\end{tabular}
\end{adjustbox}
\begin{tablenotes}[flushleft]
\footnotesize
\item These are illustrative inputs used in the experiments. They are chosen to
produce realistic-looking non-Gaussian equity dynamics, not to estimate a live
option surface or historical carry premia exactly. The companion replication
script exposes the number of paths as a parameter; higher path counts should be
used for any production statement about $\CVaR_{0.99}$.
\end{tablenotes}
\end{threeparttable}
\end{table}

The baseline parameters satisfy the Feller condition
$2\kappa\theta = 0.24 > 0.16 = \xi^2$ in continuous time, so the true
variance process remains strictly positive almost surely. The Feller
condition does not, however, carry over automatically to the discrete
scheme: a plain Euler--Maruyama step on the CIR diffusion can still
produce negative variance draws. We therefore discretize the variance
equation using the full-truncation Euler scheme of
\citet{LordKoekkoekVanDijk2010}, which replaces $\sqrt{v_t}$ in both the
diffusion and drift terms with $\sqrt{v_t^+}$ at each step; this
preserves non-negativity of the simulated variance path and has
well-documented weak convergence properties for CIR-type dynamics.
Alternative schemes with the same positivity property --- notably
Andersen's quadratic-exponential scheme --- give very similar results in
the stylized calibration used here.

\begin{figure}[tbp]
\centering
\includegraphics[width=0.96\textwidth]{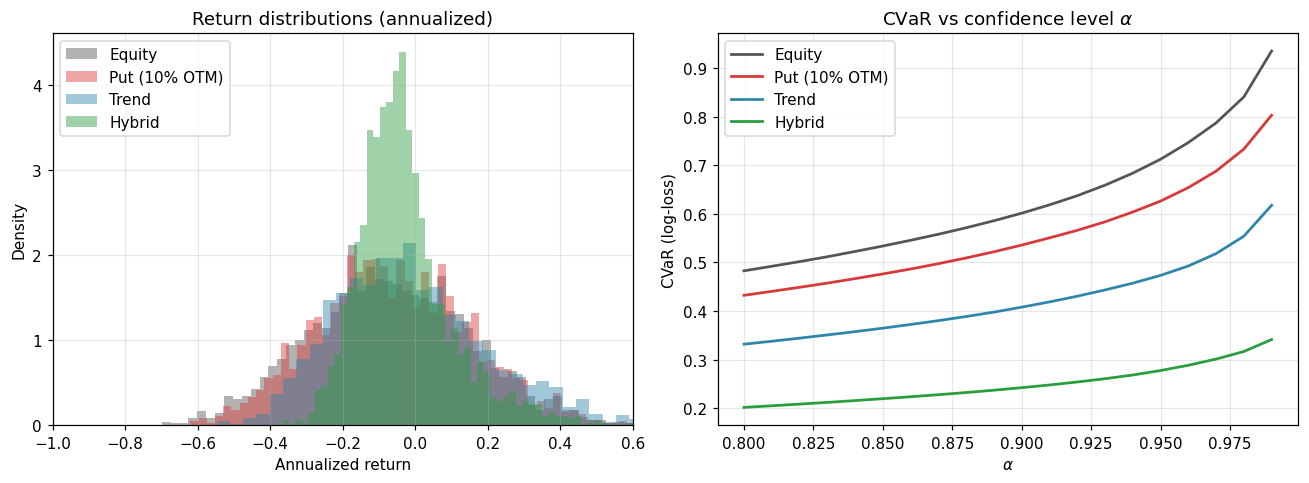}
\caption{Baseline strategy comparison from the experiments. Left: annualized return
distributions for unhedged equity, the put overlay, trend, and the hybrid.
Right: CVaR as a function of the confidence level $\alpha$. In the baseline,
the hybrid has the lowest CVaR curve across the displayed tail levels.}
\label{fig:baseline}
\end{figure}

\begin{table}[tbp]
\centering
\caption{Baseline simulation outcomes from the experiments. The last column gives
the nonparametric bootstrap standard error of the $\CVaR_{0.95}$ estimator.}
\label{tab:baseline}
\begin{threeparttable}
\scriptsize
\begin{adjustbox}{max width=\textwidth}
\begin{tabular}{lrrrrrrrr}
\toprule
Strategy & Mean & Std & Sharpe & Skew & Kurt & VaR95 & CVaR95 & SE \\
\midrule
Equity         & $-0.0673$ & $0.2240$ & $-0.300$ & $0.165$ & $-0.002$ & $0.5540$ & $0.7126$ & $0.0211$ \\
Put (10\% OTM) & $-0.0569$ & $0.2124$ & $-0.268$ & $0.238$ & $0.027$ & $0.4995$ & $0.6265$ & $0.0165$ \\
Trend          & $\phantom{-}0.0048$ & $0.2456$ & $\phantom{-}0.020$ & $1.325$ & $3.921$ & $0.3861$ & $0.4734$ & $0.0127$ \\
Hybrid         & $-0.0289$ & $0.1330$ & $-0.217$ & $1.216$ & $2.663$ & $0.2276$ & $0.2773$ & $0.0068$ \\
\bottomrule
\end{tabular}
\end{adjustbox}
\begin{tablenotes}[flushleft]
\footnotesize
\item All numbers are annualized one-year path statistics from the
experiments. The CVaR column uses the log-loss convention
$L=-\log(X_T/X_0)$. The Hybrid allocation is an equal mix ($w=0.5$)
of the Put and Trend sleeves at the per-step return level. Bootstrap
standard errors use $B=500$ resamples. Because the design is stylized,
level differences should be read qualitatively rather than as calibrated
historical performance estimates.
\end{tablenotes}
\end{threeparttable}
\end{table}

The baseline message is exactly the one suggested by Proposition~\ref{prop:interior}. The
hybrid strategy has the lowest baseline CVaR ($0.2773$), meaningfully improving
on both the pure put ($0.6265$) and pure trend ($0.4734$) overlays. Trend is the
best of the four on mean return in this stylized sample, while both put and
hybrid overlays improve left-tail metrics substantially.

\paragraph{Monte Carlo uncertainty and robustness.}
Because the baseline differences are intentionally stylized and economically modest, it is useful to report sampling uncertainty alongside point estimates. For any strategy $k$, let $\widehat{C}_{k,\alpha}$ denote the simulated $\CVaR_\alpha$. Using $B$ nonparametric bootstrap resamples of the simulated paths, define
\[
\widehat{\mathrm{se}}(\widehat{C}_{k,\alpha})
:=
\left(
\frac{1}{B-1}
\sum_{b=1}^B
\bigl(\widehat{C}_{k,\alpha}^{(b)}-\bar C_{k,\alpha}\bigr)^2
\right)^{1/2},
\qquad
\bar C_{k,\alpha}
:=
\frac{1}{B}\sum_{b=1}^B \widehat{C}_{k,\alpha}^{(b)}.
\]
Likewise, for a pair of strategies $k,\ell$, it is informative to examine the bootstrap distribution of the difference
\[
\widehat{\Delta}_{k,\ell,\alpha}
:=
\widehat{C}_{k,\alpha}-\widehat{C}_{\ell,\alpha}.
\]
This does not change the economic message of the section. It simply clarifies which conclusions are structural and which are numerically close. With $2{,}000$ paths, the $0.95$ tail contains about $100$ observations and the $0.99$ tail contains about $20$ observations, so the deepest-tail curves in the figures should be read as qualitative diagnostics rather than precise estimates. In the present design, the main object of interest is the regime ordering---put impact protection, hybrid/trend strength in persistent bears, and mixed solutions in intermediate regimes---rather than the exact magnitude of small baseline CVaR gaps.

\subsection{Four-axis hedge-quality diagnostics}
\label{sec:diag_empirical}

The diagnostic framework of \Cref{sec:diagnostics} is an interpretability
layer around the scalar CVaR objective. It is useful to verify that the four
dimensions---conditional convexity, tail reliability, non-stress carry, and
drawdown persistence---separate the strategies in the economically expected
way. \Cref{tab:diagnostics_empirical} reports the four scores computed on
weekly returns aggregated from the same Monte Carlo paths used for
\Cref{tab:baseline}. The baseline equity return stream plays the role of the
market $R_M$ in \eqref{eq:downside_beta}--\eqref{eq:persistence}. We use
$\alpha=0.95$ for the Conv/Hit/Carry diagnostics and a persistence window of
$h=4$ weeks with depth threshold $d=5\%$.

\begin{table}[tbp]
\centering
\caption{Four-axis hedge-quality diagnostics for the baseline experiment.}
\label{tab:diagnostics_empirical}
\begin{threeparttable}
\begin{adjustbox}{max width=0.95\textwidth}
\begin{tabular}{lrrrr}
\toprule
Strategy & $\mathsf{Conv}_{i,\alpha}$ & $\mathsf{Hit}_{i,\alpha}$ & $\mathsf{Carry}_{i,\alpha}$ & $\mathsf{Pers}_{i,h,d}$ \\
\midrule
Equity         & $\phantom{-}0.0002$ & $0.0000$ & $\phantom{-}0.0035$ & $0.5124$ \\
Put (10\% OTM) & $\phantom{-}0.3104$ & $0.0000^{\dagger}$ & $\phantom{-}0.0030$ & $0.4764$ \\
Trend          & $-0.1454$ & $0.6978$ & $-0.0016$ & $0.4486$ \\
Hybrid         & $\phantom{-}0.0691$ & $0.5180$ & $\phantom{-}0.0007$ & $0.3039$ \\
\bottomrule
\end{tabular}
\end{adjustbox}
\begin{tablenotes}[flushleft]
\footnotesize
\item Scores are computed on weekly returns (non-overlapping 5-day windows)
aggregated from the baseline $2{,}000$-path simulation. Conv, Hit, and Carry
use $\alpha=0.95$; Pers uses $h=4$ weeks and $d=0.05$. Under the convention
$\mathsf{Carry}_{i,\alpha}:=-\E[R_i\mid(\mathcal{A}_\alpha^0)^c]$ of
\eqref{eq:carrydiag}, a \emph{positive} Carry value indicates drag outside
stress, while a \emph{negative} value indicates a positive carry premium.
\item[$\dagger$] Scale artifact of the single-week Hit on a monthly OTM
option; see discussion below (``The zero hit ratio for Equity and Put'').
The convexity channel is captured instead by $\mathsf{Conv}_{i,\alpha}=0.31$.
\end{tablenotes}
\end{threeparttable}
\end{table}

The pattern is the expected one. The pure put sleeve has by far the highest
conditional convexity ($0.31$), consistent with its jump-repricing mechanism
$R_P(t,s,v;\zeta)$. The trend sleeve has the highest tail reliability
($0.70$): because trend goes short following a sustained drop, it is
frequently positive in market-tail weeks even though its conditional beta is
not negative. Equity has essentially zero conditional convexity by
construction ($\mathsf{Conv}_{i,\alpha}\approx 0$), since its downside and
normal-state betas to itself are both equal to one. The Hybrid strategy sits
between the two hedge corners on conditional convexity, reliability, and
carry, but its persistence score ($0.30$) sits \emph{below} both parent
sleeves and Equity rather than between them. This is consistent with the
two-channel mechanism story: the put sleeve's drawdown-window payoff is
concentrated in the early part of a decline (once the contract moves into
the money), while the trend sleeve's defensive positioning builds later,
after the signal has had time to cross through zero. A convex combination
of the two therefore produces neither a large early revaluation nor a
strongly defensive late position, so the fraction of drawdown windows with
positive cumulative return declines relative to either parent taken alone.
Importantly, this does not conflict with the CVaR ranking: $\mathsf{Pers}$
measures the probability of cumulative improvement \emph{along} a
drawdown window, whereas CVaR measures expected terminal loss in the tail.
The Hybrid optimizes the latter and is evaluated on the former; the
resulting profile is exactly what the mechanism-allocation interpretation
of hybrid demand predicts --- the optimum is not the best on any single
axis, it is the best weighted combination under the CVaR objective.

\paragraph{The zero hit ratio for Equity and Put.}
The single-week reliability score $\mathsf{Hit}_{i,\alpha}$ attains $0$ for
the Equity sleeve (as expected: being strictly equal to the market, its
return is negative whenever the market is in its own lower tail) and also
for the 10\% OTM Put sleeve. The latter is initially surprising, but
reflects the timing of option revaluation: at the moment the market enters
its lower-tail week, a 10\% OTM one-month put is typically still
out-of-the-money in Black--Scholes terms; the payoff builds only as $S$
moves further below the strike, as implied vol expands, or through the
subsequent week's repricing. Within a single 5-day window,
net-of-premium-decay sleeve returns therefore remain non-positive almost
everywhere in the lower-tail bucket. This is a known scale artifact of
one-period hit ratios for OTM option overlays rather than a statement that
the put sleeve does not help in stress; the convexity channel is captured
instead by $\mathsf{Conv}_{i,\alpha}$. A more robust reliability diagnostic
is the event-horizon variant
$\mathsf{Hit}^{h_{\mathrm{ev}}}_{i,\alpha}
:=\Prob(\prod_{k=1}^{h_{\mathrm{ev}}}(1+R_{i,t+k})>1\mid\mathcal{A}_\alpha^0)$
with $h_{\mathrm{ev}}$ equal to the roll horizon of the sleeve; computing
the event-horizon variant and cross-validating the ranking is left to a
follow-up refinement of the accompanying experiments.

The experiments then consider three regimes designed to isolate the economics
of the two hedge channels: a flash-crash regime with frequent deep negative
jumps, a prolonged-bear regime with strongly negative drift and elevated
variance, and a volatility-spike regime with high initial variance. The results
are shown in \cref{fig:scenarios,tab:scenarios}.

\begin{figure}[tbp]
\centering
\includegraphics[width=0.99\textwidth]{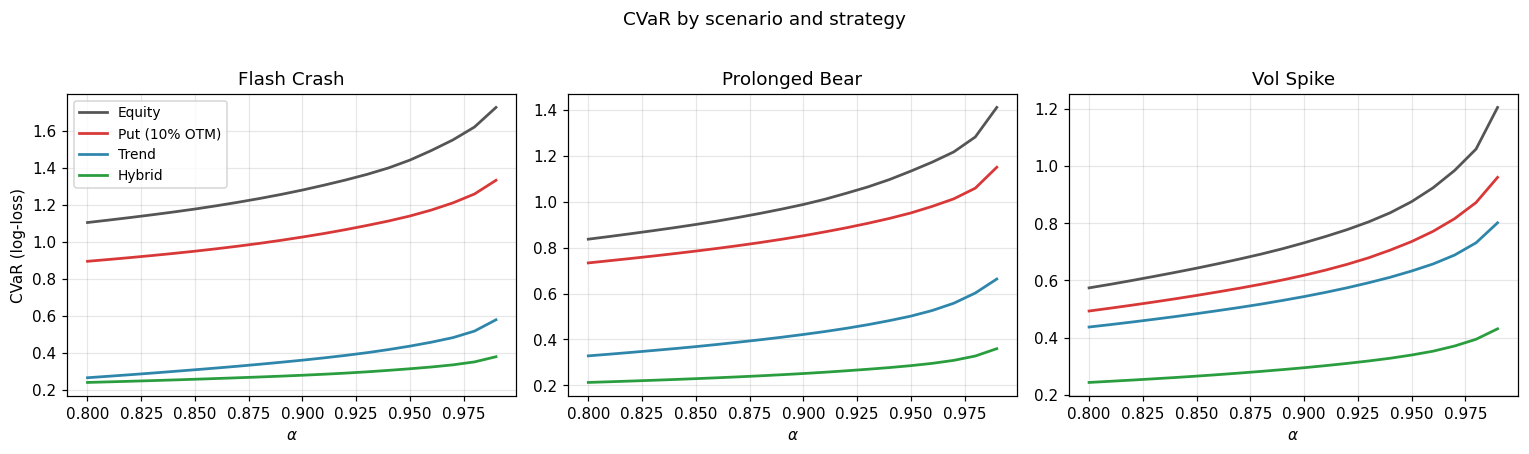}
\caption{Scenario-specific CVaR curves in the stylized experiment. In each
reported regime, the Hybrid curve lies below both pure-sleeve endpoints across
the displayed $\alpha$ levels, consistent with the mechanism logic of
Proposition~\ref{prop:interior}. Trend alone provides substantial left-tail
reduction in all three regimes; the incremental role of convex insurance is
largest in the volatility-spike regime.}
\label{fig:scenarios}
\end{figure}

\begin{table}[tbp]
\centering
\caption{Scenario-specific CVaR$_{0.95}$ and optimal put weights from the
experimental grid search.}
\label{tab:scenarios}
\begin{threeparttable}
\begin{adjustbox}{max width=0.99\textwidth}
\begin{tabular}{lrrrrrrr}
\toprule
Scenario & Equity & Put & Trend & Hybrid & $w^*_{0.90}$ & $w^*_{0.95}$ & $w^*_{0.99}$ \\
\midrule
Flash Crash    & $1.4411$ & $1.1389$ & $0.4360$ & $0.3137$ & $0.45$ & $0.50$ & $0.55$ \\
Prolonged Bear & $1.1335$ & $0.9511$ & $0.5013$ & $0.2854$ & $0.55$ & $0.55$ & $0.60$ \\
Vol Spike      & $0.8756$ & $0.7360$ & $0.6330$ & $0.3393$ & $0.60$ & $0.60$ & $0.65$ \\
\bottomrule
\end{tabular}
\end{adjustbox}
\begin{tablenotes}[flushleft]
\footnotesize
\item $w^*_{\alpha}$ is the grid-search optimum in the reduced hybrid class,
where $w=1$ corresponds to a pure put sleeve and $w=0$ to a pure trend sleeve.
The grid search uses $w\in\{0.00,0.05,\dots,1.00\}$ (spacing $\Delta w=0.05$,
$21$ points) on the same $2{,}000$-path simulation per scenario, with common
random numbers across grid points. CVaR values are on the log-loss scale
$L=-\log(X_T/X_0)$.
\end{tablenotes}
\end{threeparttable}
\end{table}

Three features stand out. First, in all three regimes the optimized Hybrid
strategy achieves the lowest CVaR within the displayed grid, with interior
optima $w^*\in[0.45,\,0.65]$ --- consistent with the sufficient-condition logic
of Proposition~\ref{prop:interior}, in the sense that the numerical objective
exhibits the interior-minimum behavior the proposition characterizes. This is an
in-simulation grid result, not a theorem of universal hybrid dominance. Its
interpretation comes from the mechanism story: puts and trend are complementary
rather than substitutable, and a combination of the two exploits both the
instantaneous-jump protection of convex insurance
(Proposition~\ref{prop:jumpprotect}) and the persistent-drawdown protection of trend
(Lemma~\ref{lem:signal_mechanics}). Second, the optimal put weight rises
monotonically with tail aversion $\alpha$ in every scenario, reflecting the
increasing marginal value of convex insurance as the investor cares about
deeper losses. Third, the relative role of puts is largest in the
volatility-spike regime, where the hybrid optimum sits near $w^*\approx
0.60\text{--}0.65$; by contrast, in the flash-crash regime the very large
Trend reduction (CVaR falls from $1.44$ to $0.44$ versus $1.14$ for the pure
put) means that Trend does most of the work and convex insurance is the
marginal complement rather than the dominant channel.

\begin{figure}[tbp]
\centering
\includegraphics[width=0.98\textwidth]{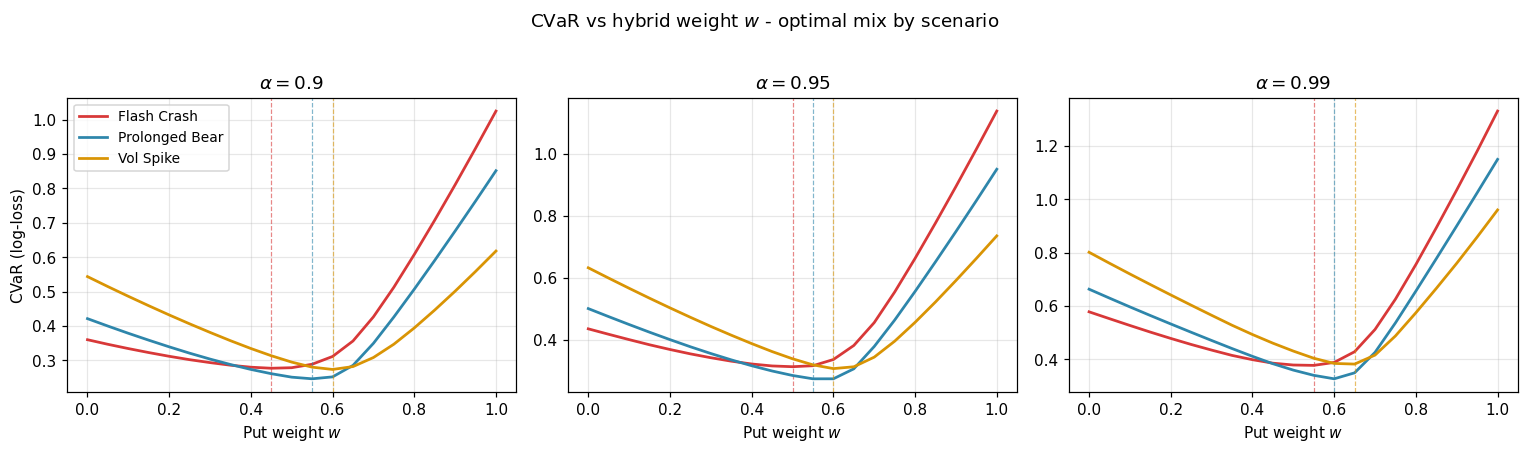}
\caption{CVaR as a function of the put weight $w$ in the reduced hybrid class,
one panel per confidence level $\alpha\in\{0.90,0.95,0.99\}$. Dashed vertical
lines mark the grid-search optimum $w^*_\alpha$ for each scenario. All curves
are visibly convex in $w$ and admit clearly interior minimizers, consistent
with the assumptions of Proposition~\ref{prop:interior}.}
\label{fig:optimalw}
\end{figure}

\subsection{Policy-gradient implementation}

The companion replication script also makes the policy-gradient result operational by optimizing the
hybrid weight directly. Using the correct tail-average gradient estimator from
\eqref{eq:pg_est}, projected gradient descent in the flash-crash regime
converges from the pure-trend corner $w_0=0$ to an interior optimum near
$w^*\approx 0.46$, and the analytic tail-average gradient matches the numerical
finite-difference gradient almost exactly across the full weight grid.

\begin{figure}[tbp]
\centering
\includegraphics[width=0.98\textwidth]{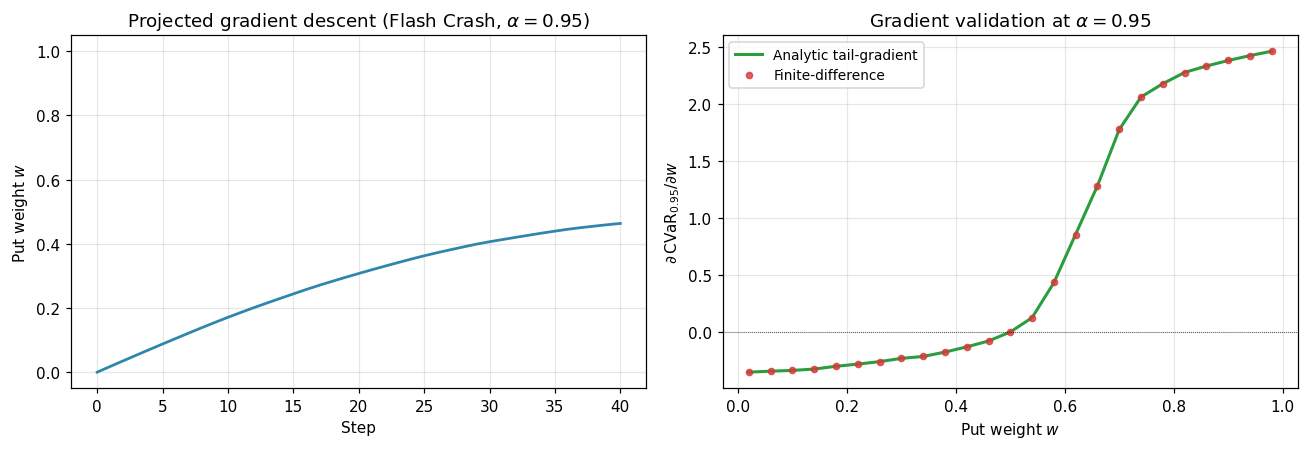}
\caption{Policy-gradient implementation in the flash-crash regime. Left:
projected gradient descent on the put weight $w$ at $\alpha=0.95$, starting
from $w_0=0$ and converging monotonically to an interior optimum near
$w^*\approx 0.46$. Right: the analytic tail-average CVaR gradient
(Theorem~\ref{thm:pg}) and a finite-difference benchmark overlap across the
full weight grid; the sign change at $w^*$ identifies the first-order
condition numerically.}
\label{fig:gradient}
\end{figure}

\FloatBarrier
\section{Discussion and limitations}
\label{sec:discussion}

The theoretical message of the paper is stronger than any one numerical
exercise. A model that respects marked-to-market option accounting, jump risk,
stochastic variance, and signal dynamics naturally produces a division of labor
between puts and trend. Puts hedge states that arrive before the signal can
react. Trend hedges states that last long enough for the signal to turn the
portfolio defensive. The hybrid question is therefore not whether one channel is
``better'' in the abstract, but which channel covers which region of the tail.

At the same time, the present implementation is intentionally modest. The
simulation is single-asset, not the multi-asset managed-futures environment in
which trend historically shines most clearly. The current option pricer is a
fast Black--Scholes proxy rather than a full Heston--jump calibration. The jump
process also leaves variance itself continuous; allowing co-jumps in price and
variance would be a natural extension. None of those limitations undermines the
main mechanism result, but they do matter for production use.

The practical implication is that a live allocation would require an additional
validation layer: option-surface calibration, transaction-cost and bid--ask
stress tests, independent holdout evaluation of the selected hybrid weight, and
confidence intervals for all reported CVaR improvements. The framework is meant
to organize that validation, not to replace it.

A fully empirical version of the paper would extend the present framework in
three directions. First, estimate the overlay drift $\mu_P$ and jump repricing
term $R_P$ from a live index option surface. Second, move from a single-asset
trend signal to the diversified multi-asset construction emphasized by
\citet{Hurst2017} and \citet{Ilmanen2020}. Third, allow the crash intensity
$\lambda$ to be regime dependent or self-exciting so that convex demand can rise
endogenously into stressed states; the mutually-exciting jump-process
framework of \citet{AitSahalia2015} for financial contagion is a natural
starting point, and connects directly to empirical evidence that tail-risk
premia concentrate in clustered episodes rather than independent arrivals.

\subsection{Extension to a tradable hedge universe}
\label{sec:tradeable_universe}

The two-channel model should also be viewed as the smallest member of a larger
strategy-vector formulation. Let $H_t=(H_t^1,\ldots,H_t^n)$ denote tradable hedge
sleeves, such as listed index options, put spreads, variance or volatility
exposures, dispersion trades, volatility futures, cross-asset trend, rates trend,
FX trend, commodity trend, or other liquid relative-value hedges. The wealth
equation becomes
\begin{equation}
\frac{dX_t}{X_{t^-}}
= \pi_t\frac{dS_t}{S_{t^-}}
+ w_t^\top \frac{dH_t}{H_{t^-}},
\qquad w_t\in\mathcal{W}\subset\R^n,
\label{eq:vector_wealth}
\end{equation}
where $\mathcal{W}$ contains budget, leverage, liquidity, and shorting
constraints. Each sleeve can be assigned a characteristic vector
\begin{equation}
z_i
=
\left(
\mathsf{Conv}_{i,\alpha},
\mathsf{Hit}_{i,\alpha},
\mathsf{Carry}_{i,\alpha},
\mathsf{Pers}_{i,h,d}
\right).
\label{eq:characteristic_vector}
\end{equation}
The structuring problem is then not simply to choose instruments, but to choose a
portfolio of mechanisms. For example, a volatility sleeve may have strong
conditional convexity but high carry cost; a cross-asset trend sleeve may have
weaker impact protection but better persistence; and a dispersion sleeve may
supply a more specific volatility-risk exposure. The optimization can select
weights $w$ directly while imposing the constraints in
\eqref{eq:constrained_diagnostic} or the penalties in
\eqref{eq:penalized_diagnostic}. This turns the model into a design tool: a
structured hedge can be described not only by its CVaR improvement, but by the
explicit mix of convexity, reliability, carry, and persistence it delivers.

The present paper does not estimate this full characteristic matrix. Doing so
would require historical or surface-implied returns for each sleeve and common
stress definitions across asset classes. But the mathematical extension is
straightforward: replace the scalar convex exposure $q_t$ by $w_t$, augment the
state by any sleeve-specific signals or volatility factors needed for Markovian
dynamics, and solve the same CVaR control problem on the enlarged state space.

\paragraph{What is structural and what is calibration-dependent.}
The strongest conclusions of the paper are analytical rather than numerical.
The marked-to-market accounting correction, the state augmentation to
$(X,S,v,M)$, and the separation of two protection mechanisms --- immediate
jump repricing for convex insurance and delayed signal-driven derisking for
trend --- all arise from the model structure itself. By contrast, the exact
magnitudes of the simulated CVaR gaps, the location of the hybrid optimum,
and the quantitative slope of the weight-CVaR curve are calibration-dependent.
This distinction is important for interpretation. The simulations should be
read as computational illustrations of the structural ranking, not as a claim
that the reported weight levels or tail improvements are stable empirical
constants across markets, maturities, or option-surface environments.

\section{Conclusion}
\label{sec:conclusion}

This paper develops a continuous-time framework for tail-risk management that puts
convex option insurance and dynamic trend-following inside one coherent CVaR
objective. The key modeling choice is to treat the option overlay as a traded
marked-to-market asset. Once that accounting is imposed, the relevant state is
not just wealth but wealth, spot, variance, and the trend signal; and once that
state is written down, the economic separation of the two protection
mechanisms becomes sharp.

The central mechanism story is about timing rather than dominance. On an
isolated jump, convex insurance reprices contractually and immediately, while
trend is necessarily late because its signal must first cross through zero
(Proposition~\ref{prop:jumpprotect} and Corollary~\ref{cor:threshold}). Over a horizon that
contains repeated jumps or a sustained drawdown, however, trend has time to
react: after each shock the signal drifts toward negative values and the
overlay turns defensive (Lemma~\ref{lem:signal_mechanics}). A year-long flash-crash
\emph{regime} is therefore not the same object as one jump on impact: it
contains many between-jump stretches in which trend can work, and the optimal
hedge over that regime is not the same as the optimal hedge against a single
shock.

The stylized experimental evidence reflects exactly this picture. In the
reported simulations, hybrid allocations deliver the lowest CVaR within the
specified strategy grid, with interior optima near
$w^*\in[0.45,\,0.65]$ across baseline, flash-crash, prolonged-bear, and
vol-spike calibrations. Trend accounts for a large share of the CVaR reduction
in regimes with repeated jumps because it has time to react between shocks,
while convex insurance adds a further improvement on top of trend and becomes
relatively more valuable as the tail-aversion level $\alpha$ rises. The
diagnostic layer explains why both channels are wanted: the put sleeve loads
strongly on conditional convexity, the trend sleeve loads strongly on tail
reliability, and the hybrid is the allocation that balances these mechanisms
against carry and persistence. That is the kind of state-contingent answer
practitioners need. Tail-risk management is not a contest between convex
insurance and trend. It is an allocation problem across loss mechanisms, and
in a production setting those mechanisms can be represented as a vector of
tradable strategies rather than a small set of stylized instruments.

\appendix

\section{Replication and validation checklist}
\label{app:replication}

The accompanying package includes a Python script,
\path{replication_tail_risk_management.py}, that implements the same
stylized simulation architecture used in \Cref{sec:mc}: Heston-type variance
with full truncation, finite-activity jumps, rolled Black--Scholes put proxies,
EWMA trend signals, hybrid weights, CVaR curves, bootstrap standard errors, and
a projected-gradient illustration. The script is intentionally transparent
rather than optimized. Its purpose is to make the numerical pipeline auditable
and to let the reader increase the path count, alter the option-pricing proxy,
or introduce transaction-cost assumptions.

For a production empirical study, the following validation steps should be
reported alongside any headline CVaR improvement:
\begin{enumerate}[leftmargin=1.6em]
    \item \textbf{Common-random-number grid.} All hybrid weights should be
    evaluated on the same simulated shocks so that differences across weights
    reflect policy changes rather than Monte Carlo noise.
    \item \textbf{Holdout selection.} The weight $w^*$ should be selected on a
    training batch or historical window and evaluated on an independent holdout
    batch or later historical window.
    \item \textbf{Tail uncertainty.} CVaR estimates should be accompanied by
    bootstrap or asymptotic standard errors. At confidence level $\alpha$, an
    $N$-path simulation uses roughly $(1-\alpha)N$ tail observations; this is
    especially important for $\alpha=0.99$.
    \item \textbf{Option-surface sensitivity.} The put sleeve should be stress
    tested across implied-volatility premia, skew assumptions, bid--ask spreads,
    and roll-date conventions.
    \item \textbf{Mechanism diagnostics.} The four scores
    $(\mathsf{Conv},\mathsf{Hit},\mathsf{Carry},\mathsf{Pers})$ should be
    reported together with CVaR so that a lower tail loss is not mistaken for a
    single economic source of protection.
\end{enumerate}

These checks separate the structural claim of the paper---convex crash
protection and signal-driven drawdown protection are complementary mechanisms---
from the calibration-specific claim that a particular hybrid weight is optimal
in a particular market environment.

\printbibliography

\end{document}